\documentclass[twoside,11pt]{article}

\usepackage{blindtext}

\usepackage{jmlr2e}

\usepackage{amsmath}
\usepackage{enumerate}
\usepackage{url}
\usepackage{bm}
\usepackage{algorithm}
\usepackage{algpseudocode}
\usepackage{float}
\usepackage{multirow}
\usepackage{subfigure}
\usepackage{doi}
\usepackage[dvipsnames, svgnames, x11names]{xcolor}

\newcommand{\leftsecond}{\left[}
\newcommand{\rightsecond}{\right]}
\newcommand{\leftthird}{\left\{}
\newcommand{\rightthird}{\right\}}
\newcommand{\leftdot}{\left.}
\newcommand{\rightdot}{\right.}

\usepackage{lastpage}
\jmlrheading{00}{0000}{1-\pageref{LastPage}}{; Revised }{}{}{Qihua Wang, Jinye Du and Ying Sheng}

\ShortHeadings{Distributed Empirical Likelihood}{Wang, Du and Sheng}
\firstpageno{1}

\begin{document}

\title{Distributed Empirical Likelihood Inference With or Without Byzantine Failures}

\author{\name Qihua Wang \email qhwang@amss.ac.cn \\
       \addr Academy of Mathematics and Systems Science, Chinese Academy of Sciences\\
       University of Chinese Academy of Sciences\\
       Zhongguancan East Road, Beijing 100190, China
       \AND
       \name Jinye Du \email dujinye14@mails.ucas.ac.cn \\
       \addr Academy of Mathematics and Systems Science, Chinese Academy of Sciences\\
       University of Chinese Academy of Sciences\\
       Zhongguancan East Road, Beijing 100190, China
       \AND
       \name Ying Sheng \email shengying@amss.ac.cn \\
       \addr Academy of Mathematics and Systems Science, Chinese Academy of Sciences\\
       University of Chinese Academy of Sciences\\
       Zhongguancan East Road, Beijing 100190, China}

\editor{My editor}

\maketitle

\begin{abstract}
Empirical likelihood is a very important nonparametric approach which is of wide application. However, it is hard and even infeasible to calculate the empirical log-likelihood ratio statistic with massive data. The main challenge is the calculation of the Lagrange multiplier. This motivates us to develop a distributed empirical likelihood method by calculating the Lagrange multiplier in a multi-round distributed manner. It is shown that the distributed empirical log-likelihood ratio statistic is asymptotically standard chi-squared under some mild conditions. The proposed algorithm is communication-efficient and achieves the desired accuracy in a few rounds. Further, the distributed empirical likelihood method is extended to the case of Byzantine failures. A machine selection algorithm is developed to identify the worker machines without Byzantine failures such that the distributed empirical likelihood method can be applied. The proposed methods are evaluated by numerical simulations and illustrated with an analysis of airline on-time performance study and a surface climate analysis of Yangtze River Economic Belt.
\end{abstract}


\section{Introduction}\label{sec:intro}

Empirical likelihood (EL) was first introduced by \citet{owen1988empirical,owen1990empirical} as a nonparametric approach for constructing confidence regions. It has many nice properties: automatic determination of the shape of confidence regions \citep{owen1988empirical}, incorporation of auxiliary information through constraints \citep{qin1994empirical}, and range respecting and transformation preserving \citep{owen2001empirical}.
EL has been studied extensively. See, e.g., \cite{hall1990methodology,owen1991empirical,diciccio1991empirical,chen1993empirical,chen1993smoothed,qin1993empirical,qin1995estimating,kitamura1997empirical}.
As a recognized powerful tool in statistical inference, the use of EL has become very popular in various fields, including missing data analysis \citep{wang2002empirical, wang2009empirical}, high-dimensional data analysis \citep{tang2010penalized,lahiri2012penalized, leng2012penalized, chang2018new,chang2021high}, censored data analysis \citep{li2003empirical,he2016empirical,tang2020penalized}, longitudinal data analysis \citep{xue2007empirical,wang2010generalized,shi2011two}, measurement error analysis \citep{wang2002biome}, meta-analysis \citep{qin2015using,huang2016efficient,han2019empirical,zhang2020generalized}, information integration \citep{ma2022statistical}, and so on.
For a comprehensive review of EL, see \cite{owen2001empirical}, \cite{chen2009review}, and \cite{lazar2021review}.

In recent years, unprecedented technological advances in data generation and acquisition have led to a proliferation of massive data, posing new challenges to statistical analysis. When the sample size is extremely large, it is computationally infeasible to perform standard statistical analysis on a single machine due to memory limitations. Even for a single machine with sufficient memory, optimization algorithms with a massive amount of data are computationally expensive, leading to unaffordable time costs. To reduce the cost of computation, there has been a growing interest in developing distributed statistical approaches in recent years.
The main strategy of distributed statistical approaches is to divide the entire data set into several subsets, calculate the local statistic using each subset in each worker machine, and aggregate local statistics from worker machines into a summary statistic.
The distributed statistical approaches have drawn much attention in various areas, including distributed  M-estimation \citep{shi2018massive,jordan2019communication,fan2021communication}, distributed empirical likelihood estimators \citep{liu2023distributed,zhou2023distributed},   principal component analysis \citep{fan2019distributed,chen2021distributed,huang2021communication}, high-dimensional test and estimation \citep{lee2017communication,battey2018distributed,hector2021distributed}, quantile regression \citep{chen2019quantile,volgushev2019distributed}, support vector machine \citep{lian2018divide,wang2019distributed}, and so on.

To develop the distributed empirical likelihood method, a simple method is to calculate the global empirical log-likelihood with the Lagrange multiplier calculated by the average of local Lagrange multipliers from worker machines.
Another natural method is the split sample empirical likelihood approach \citep{jaeger2020split}, which calculates the global empirical log-likelihood by taking the summation of local empirical log-likelihood from worker machines.
However, the empirical log-likelihood ratio statistics obtained by the aforementioned two approaches are not asymptotically standard chi-squared when the number of worker machines diverges.
In practice, the number of machines is usually divergent with massive data. Otherwise, the size of the sample in every machine must be of the same order as the size of the massive data set.
Recently, \cite{zhou2023distributed} has developed a distributed empirical likelihood method.
However, they use the alternating direction method of multipliers, which introduces superfluous parameters in calculation. 
This results in a slow convergence rate and heavy computation load.
Additionally, it imposes strict restrictions on the initial values and the number of worker machines, which is often impractical.
The number of machines is restricted to diverge at a very slow rate such that the sample size in every 
machine must be at least of the same order with $N^{2/3}$,  where $N$ is the whole sample size.
This yields almost the same calculation and data storage problems as the case of a single machine when $N$ is massive, and hence it does not make sense in practice.


The existing literature on distributed empirical likelihood has not made an essential advance.
This motivates us to develop a distributed empirical likelihood (DEL) method, which allows the number of machines to diverge at a fast rate such that every machine can save the data allocated and make calculation in a short time even if $N$ is massive.
The literature uses average or weighted average to aggregate the local empirical likelihood statistics.
We use a different technique.
By convex dual representation \citep{owen2001empirical}, the Lagrange multiplier can be calculated by minimizing a global optimal function, which is the average of local optimal functions. The naive method, which takes the average of the local Lagrange multipliers calculated by minimizing local optimal functions, is not applicable since the average does not approximate the global Lagrange multiplier sufficiently well and the resulting distributed empirical log-likelihood ratio statistic is not asymptotically standard chi-squared when the number of machines diverges.
Similar comments also apply to the split empirical likelihood approach. To tackle this problem, we construct a modified local optimal function by replacing the first derivative in Taylor's expansion of the local optimal function in each worker machine with the global one.
The Lagrange multiplier is then calculated using a multi-round distributed algorithm. In each round, we calculate the aggregated Lagrange multiplier by averaging local Lagrange multipliers, which are obtained by minimizing the modified local optimal functions.
The proposed DEL method can not only simultaneously achieve high statistical accuracy and low computation cost, but also inherit the nice properties of EL.
Furthermore, it is shown that the distributed empirical log-likelihood ratio statistic is asymptotically standard chi-squared under some mild moment conditions.

It is worth noting that Byzantine failures may occur for distributed statistical methods, where the information sent from a worker machine can be arbitrarily erroneous due to hardware or software breakdowns, data crashes, or communication failures \citep{lamport1982byzantine}.
The definition of Byzantine failures in the mathematical form can be seen in \cite{tu2021variance}.
The Byzantine failures, if not addressed properly, may lead to invalid statistical inferences since the averaged global gradient can be completely skewed by some worker machines with Byzantine failures.
In recent years, various approaches have been proposed for statistical learning and inference with Byzantine failures \citep{blanchard2017machine, alistarh2018byzantine, yin2018byzantine, tu2021variance}. Instead of using the vanilla gradient mean, existing methods aggregate gradients from worker machines using some robust mean estimators, such as the trimmed mean \citep{yin2018byzantine}, the median of mean \citep{tu2021variance}. Different from these existing approaches, our goal is to identify the worker machines without Byzantine failures, and fully utilize gradient information from such machines. To this end, we first propose a machine selection algorithm to identify the worker machines which return the correct gradient information. The proposed DEL method can then be applied to the selected worker machines. Under some regularity conditions, we establish the selection consistency for the proposed machine selection procedure and prove that the corresponding distributed empirical log-likelihood ratio statistic with the selected machines is asymptotically standard chi-squared.

The paper is organized as follows.
In Section \ref{Sec:DEL}, we develop DEL for mean inference with massive data and prove that the distributed empirical log-likelihood ratio statistic is asymptotically standard chi-squared.
In Section \ref{Sec:RDEL}, we extend DEL to the case of Byzantine failures by developing a machine selection algorithm to select machines without Byzantine failures. It is proved that the algorithm possesses the selection consistency property and the corresponding distributed empirical log-likelihood ratio statistic is asymptotically standard chi-squared.
In Section \ref{simulations}, we conduct some simulation studies to demonstrate the performance of the proposed methods.
An analysis of airline on-time performance study and a surface climate analysis of Yangtze River Economic Belt are given in Section \ref{realdata}.

\section{Distributed Empirical Likelihood Inference Without Byzantine Failures}\label{Sec:DEL}

To illustrate the proposed DEL method, we begin with the statistical inference of the mean. Let $\bm Z_1,\bm Z_2,\ldots,\bm Z_N \in \mathbb R^d$ be independent and identically distributed random vectors.
Throughout this paper, we assume that $\operatorname{E} (\bm Z_1) =\bm\mu_0$.
By \cite{owen2001empirical}, if the convex hull of $\{\bm Z_1,\ldots,\bm Z_N\}$ contains $\bm\mu\in\mathbb{R}^d$, the empirical log-likelihood ratio statistic is
\begin{equation}\label{eq:ellmu}
\ell(\bm\mu) = 2\sum_{s=1}^{N} \log \leftsecond 1+ \widehat{\bm\lambda}^{*\top}(\bm Z_s-\bm\mu)\rightsecond,
\end{equation}
where $\widehat{\bm\lambda}^*$ is the unique solution of
$$
\frac{1}{N}\sum_{s=1}^{N} \frac{\bm Z_s-\bm\mu}{1+\bm\lambda^\top (\bm Z_s-\bm\mu)} = \bm 0.
$$
The calculation of $\widehat{\bm\lambda}^*$ is time costly and requires large memory capacity when sample size $N$ is extraordinarily large.
However, it is quite challenging to develop a distributed method for obtaining an approximation of $\widehat{\bm\lambda}^*$ such that $\ell(\bm\mu)$ with $\widehat{\bm\lambda}^*$ replaced by the approximation is asymptotically standard chi-squared.

It is noted that $\widehat{\bm\lambda}^*=\arg\min_{\bm\lambda} g(\bm\lambda;\bm\mu)$ by the convex dual representation, where $g(\bm\lambda;\bm\mu) =  -N^{-1}\sum_{i=1}^{N}\log \leftsecond 1+\bm\lambda^\top (\bm Z_i-\bm\mu) \rightsecond$ is the global optimal function.
To develop a distributed empirical likelihood method, we consider the distributed setting where $\bm Z_1,\ldots,\bm Z_N$ are stored in $K$ worker machines with the equal sample size $n=N/K$ and denote by $\bm X_{i,j}$ the $j$-th ($j=1,\ldots,n$) observation in the $i$-th ($i=1,\ldots,K$) worker machine.
For $i=1,\ldots,K$, we define the $i$-th local optimal function with the data stored in the $i$-th worker machine by $g_i (\bm\lambda;\bm\mu ) = -n^{-1}\sum_{j=1}^{n}\log \leftsecond 1+ \bm\lambda^\top(\bm X_{i,j} - \bm\mu) \rightsecond$.
The global optimal function can be represented as $g(\bm\lambda;\bm\mu) = K^{-1}\sum_{i=1}^K g_i(\bm\lambda;\bm\mu)$. Hereinafter, we denote $g(\bm\lambda;\bm\mu)$ and $g_i(\bm\lambda;\bm\mu)$ as $g(\bm\lambda)$ and $g_i(\bm\lambda)$ whenever no confusion arises.

The naive method calculates the Lagrange multiplier by the average $\bm{\bar\lambda}^* = K^{-1}\sum_{i=1}^{K} \bm{\widehat\lambda}_i^*$, where $\bm{\widehat\lambda}_i^*=\arg\min_{\bm\lambda} g_i(\bm\lambda)$ for $i=1,\ldots,K$.
However, $\|\bm{\bar\lambda}^* - \bm{\widehat\lambda}^* \|_2$ cannot attain $o_p(N^{-1/2})$ with a diverging $K$. As a consequence, the corresponding empirical log-likelihood ratio statistic is not asymptotically standard chi-squared.
Recalling the definition of $\widehat{\bm\lambda}^*_i$, the main reason is that the minimizer of $g_i(\bm\lambda)$ is not close to that of $g(\bm\lambda)$ sufficiently. This motivates us to define a modified local optimal function whose minimizer is closer to that of $g(\bm\lambda)$ by modifying $g_i(\bm\lambda)$ for $i=1,\ldots,K$. In what follows, let us give the details.
Given an initial estimator $\widehat{\bm\lambda}_0$, by Taylor's expansions of $g(\bm\lambda)$ and $g_i(\bm\lambda)$ at $\widehat{\bm\lambda}_0$, it yields
\begin{equation}\label{eq:g}
g(\bm\lambda) = g(\widehat{\bm\lambda}_0) + \nabla g(\widehat{\bm\lambda}_0)^\top(\bm\lambda - \widehat{\bm\lambda}_0) + R(\bm\lambda;\widehat{\bm\lambda}_0),
\end{equation}
and
\begin{equation}\label{eq:gi}
g_i(\bm\lambda) = g_i(\widehat{\bm\lambda}_0) + \nabla g_i(\widehat{\bm\lambda}_0)^\top(\bm\lambda - \widehat{\bm\lambda}_0) + R_i(\bm\lambda;\widehat{\bm\lambda}_0),
\end{equation}
where $R(\bm\lambda;\widehat{\bm\lambda}_0) = \sum_{k=2}^{\infty} \nabla^k g(\widehat{\bm\lambda}_0) (\bm\lambda - \widehat{\bm\lambda}_0)^{\otimes k}/k!$ and $R_i(\bm\lambda;\widehat{\bm\lambda}_0) = \sum_{k=2}^{\infty} \nabla^k g_i(\widehat{\bm\lambda}_0) (\bm\lambda - \widehat{\bm\lambda}_0)^{\otimes k}/k!$ are the corresponding remainders of higher-order derivatives of $g(\bm\lambda)$ and $g_i(\bm\lambda)$, respectively. Clearly, both the minimizers of $g_i(\bm\lambda)$ and $g(\bm\lambda)$ are determined mainly by the linear term on the right hand side of \eqref{eq:g} and \eqref{eq:gi}.
In order to make the minimizers of $g_i(\bm\lambda)$ closer to that of $g(\bm\lambda)$ sufficiently, we define the modified local optimal functions
\begin{equation}\label{eq:checkgi}
\check g_i(\bm\lambda;\widehat{\bm\lambda}_0) := g_i(\widehat{\bm\lambda}_0) + \nabla g(\widehat{\bm\lambda}_0)^\top(\bm\lambda - \widehat{\bm\lambda}_0) + R_i(\bm\lambda;\widehat{\bm\lambda}_0)
\end{equation}
by replacing $\nabla g_i(\widehat{\bm\lambda}_0)$ with $\nabla g(\widehat{\bm\lambda}_0)$ in \eqref{eq:gi} for $i=1,\ldots,K$.
From \eqref{eq:gi}, we have
\begin{equation}\label{eq:Ri}
R_i(\bm\lambda;\widehat{\bm\lambda}_0)=g_i(\bm\lambda) - g_i(\widehat{\bm\lambda}_0) - \nabla g_i(\widehat{\bm\lambda}_0)^\top(\bm\lambda - \widehat{\bm\lambda}_0).
\end{equation}
Take
\begin{equation}\label{eq:tildegi2}
\widetilde g_i(\bm\lambda;\widehat{\bm\lambda}_0) = g_i(\bm\lambda) + \leftsecond \nabla g(\widehat{\bm\lambda}_0) - \nabla g_i(\widehat{\bm\lambda}_0) \rightsecond^\top \bm\lambda, ~i=1,\ldots,K,
\end{equation}
by substituting \eqref{eq:Ri} into \eqref{eq:checkgi} and then omitting the constant term which does not depend on $\bm\lambda$. $\widetilde g_i(\bm\lambda;\widehat{\bm\lambda}_0)$ is equivalent to $\check g_i(\bm\lambda;\widehat{\bm\lambda}_0)$ in \eqref{eq:checkgi} in the sense that they have the same minimizer.
By \eqref{eq:tildegi2}, the minimizer of $\widetilde g_i(\bm\lambda; \widehat{\bm\lambda}_0)$ depends on $\widehat{\bm\lambda}_0$ for $i=1,\ldots,K$. In practice, it is infeasible to obtain an initial estimator $\widehat{\bm\lambda}_0$ such that the Lagrange multiplier $\widehat{\bm \lambda}_1 = K^{-1}\sum_{i=1}^K \arg\min_{\bm\lambda} \widetilde g_i(\bm\lambda;\widehat{\bm\lambda}_0)$ satisfies $\| \widehat{\bm \lambda}_1 -\widehat{\bm\lambda}^* \|_2 =o_p(N^{-1/2})$. Therefore, the empirical log-likelihood ratio statistic in \eqref{eq:ellmu} with $\widehat{\bm\lambda}^*$ replaced by $\widehat{\bm \lambda}_1$ is generally not asymptotically standard chi-squared. To tackle this problem, we develop a multi-round distributed algorithm to calculate the Lagrange multiplier in Algorithm \ref{alg:EL}.
The Lagrange multiplier after $T$ rounds is denoted by $\widehat{\bm\lambda}_T$.
Under some regularity conditions, it can be proved that the empirical log-likelihood ratio statistic is still asymptotically standard chi-squared with $\widehat{\bm\lambda}^*$ replaced by $\widehat{\bm\lambda}_T$ as long as $T$ is large enough.
In what follows, we present the algorithm.

\begin{algorithm}[htbp]
\caption{Multi-round distributed empirical likelihood method}
\label{alg:EL}
\begin{algorithmic}[1]
\Require
An initial estimator $\widehat{\bm\lambda}_0$, the number of round $T$.

\State Set $t=0$.

\For{$t=0,1,\ldots,T-1$}

\State For $i=1,\ldots,K$, the $i$-th worker machine computes the local gradient $\nabla g_i(\widehat{\bm\lambda}_t)$ and sends it to the central machine;

\State The central machine computes the global gradient
$$
\nabla g(\widehat{\bm\lambda}_t) = \frac{1}{K} \sum_{i=1}^{K} \nabla g_i(\widehat{\bm\lambda}_t),
$$
and broadcasts it to each worker machine;

\State For $i=1,2\ldots,K$, the $i$-th worker machine computes
\begin{equation*}
\widehat{\bm\lambda}_{t+1}^{(i)} = \arg\min_{\bm\lambda} \leftthird g_i(\bm\lambda) + \leftsecond \nabla g(\widehat{\bm\lambda}_t) - \nabla g_i(\widehat{\bm\lambda}_t) \rightsecond^\top \bm\lambda \rightthird,
\end{equation*}
and sends it to the central machine;

\State The central machine computes $\widehat{\bm\lambda}_{t+1} = K^{-1}\sum_{i=1}^{K}\widehat{\bm\lambda}_{t+1}^{(i)}$ and broadcasts it to each worker machine.

\EndFor

\Ensure
$\widehat{\bm\lambda}_{T}$.
\end{algorithmic}
\end{algorithm}

The initial estimator $\widehat{\bm\lambda}_0$ can be set as the zero vector or taken to be $\widehat{\bm\lambda}_0=\arg\min_{\bm\lambda} g_1(\bm\lambda)$.
The Lagrange multiplier after $T$ rounds achieves the desired accuracy, that is $\|\widehat{\bm\lambda}_T - \widehat{\bm\lambda}^*\|_2 = o_p(N^{-1/2})$, as long as $T$ is large enough.
Given $\widehat {\bm \lambda}_T$, the distributed empirical log-likelihood ratio statistic can be constructed by
\begin{eqnarray}
\label{ellDEL}
\ell_{DEL}(\bm\mu)=2\sum_{i=1}^{K}\sum_{j=1}^{n} \log \leftsecond 1+\widehat{\bm\lambda}_T^\top(\bm X_{i,j}-\bm\mu)\rightsecond.
\end{eqnarray}
To establish the asymptotic distribution of $\ell_{DEL}(\bm\mu)$, we first prove the following lemma, which presents some nice properties of the global optimal function $g(\bm\lambda)$.

\begin{lemma}\label{Lemma:assumption}
Assume the eigenvalues of $\operatorname{Cov}(\bm X_{1,1})$ are bounded away from zero and infinity, $\operatorname{E}\| \bm X_{1,1}\|_2^\beta<\infty$ for $\beta \ge 4$, and $K=O(N^{1-2/\beta})$.
Let $\mathcal B(\bm \alpha, r) := \{ \bm \lambda \in \mathbb R^d :\, \|\bm \lambda - \bm \alpha \|_2 \le r\}$ for some vector $\bm \alpha\in \mathbb R^d$ and constant $r > 0$. As $n\to\infty$, the following conclusions hold with probability tending to $1$:

(a) (Strong convexity) There exists some constant $C_1>0$ such that $\nabla^2g(\bm\lambda)\succeq \tau \mathbf I_d$ for $\bm\lambda \in \mathcal B(\widehat{\bm\lambda}^*,C_1 n^{-1/2})$, where $\tau>0$ is a constant, $\mathbf I_d$ is the $d\times d$ identity matrix, and $\mathbf P\succeq \mathbf Q$ means $\mathbf P-\mathbf Q$ is positive semi-definite for matrices $\mathbf P$ and $\mathbf Q$.

(b) (Homogeneity) There exists some constant $C_2>0$ such that $\| \nabla^2 g_i(\bm\lambda) - \nabla^2 g(\bm\lambda) \|_2 \leq C_2 n^{-1/2} \sqrt{\log K}$ uniformly for $i=1,\ldots,K$ and $\bm\lambda\in \mathcal B(\widehat{\bm\lambda}^*,C_1 n^{-1/2})$.

(c) (Smoothness of Hessian) There exists some constant $M>0$ such that $\| \nabla^2 g(\bm\lambda_1) - \nabla^2 g(\bm\lambda_2) \|_2 \leq M \|\bm\lambda_1 - \bm\lambda_2 \|_2$ for $ \bm\lambda_1, \bm\lambda_2 \in \mathcal B(\widehat{\bm\lambda}^*, C_1 n^{-1/2})$.
\end{lemma}

The strong convexity property ensures that the global optimal function $g(\bm \lambda)$ has a unique minimizer.
The homogeneity property ensures that the difference between $R(\bm\lambda;\widehat{\bm\lambda}_0)$ and $R_i(\bm\lambda;\widehat{\bm\lambda}_0)$ is small sufficiently such that $\widehat{\bm\lambda}_T$ can be sufficiently close to $\widehat{\bm\lambda}^*$.
The smoothness of hessian ensures that Algorithm \ref{alg:EL} has a faster contraction rate.

\begin{theorem}\label{thm:DEL}
Under the conditions of Lemma \ref{Lemma:assumption}, if $\|\widehat{\bm\lambda}_0 - \widehat{\bm\lambda}^*\|_2 = O_p(n^{-1/2})$ and $T \geq \lfloor \log K/\log n \rfloor + 1$, we have

(a) $ \|\widehat{\bm\lambda}_T - \widehat{\bm\lambda}^*\|_2=o_p(N^{-1/2});$

(b) $\ell_{DEL}(\bm \mu_0) \stackrel{d}\rightarrow \chi_{(d)}^2$ as $n\to\infty$,  where  $\ell_{DEL}(\bm\mu_0)$ is defined by \eqref{ellDEL} with $\bm\mu=\bm \mu_0$ and ``$\overset{d}{\to}$'' denotes convergence in distribution.
\end{theorem}

It is worth noting that if $K\le n$, $ \hat{ \lambda}_T$ asymptotically achieves the desired accuracy for $T=2$ according to Theorem \ref{thm:DEL}, which ensures that $\ell_{DEL}(  \mu_0)$ is asymptotically standard chi-squared. Therefore, $\ell_{DEL}(\mu_0)$ can be applied to testing the null hypothesis ``H$_0$: $\mu=\mu_0$''.
H$_0$ is rejected if $\ell_{DEL}(\mu_0) > \chi^2_{(d)}(\alpha)$ at $\alpha$ level, where $0<\alpha<1$ and $\chi^2_{(d)}(\alpha)$ is the $\alpha$ quantile of the standard chi-squared distribution with $d$ degrees of freedom. Moreover, $\ell_{DEL}(\mu)$ can be used to construct $1-\alpha$ confidence region (interval) $\{\mu: \ell_{DEL}(\mu) \le \chi^2_{(d)}(\alpha)\}$.

\section{Distributed Empirical Likelihood Inference With Byzantine Failures}\label{Sec:RDEL}

An important operation in the proposed DEL algorithm is that the central machine receives the transmitted gradient information from worker machines and then aggregates the local gradients by taking the average.
However, in practice, the information sent from a worker machine can be arbitrarily erroneous due to hardware or software breakdowns, data crashes, or communication failures \citep{tu2021variance}.
The success of Algorithm \ref{alg:EL} is based on the fact that the local gradients from all worker machines obtained through transmission are correct, which cannot be guaranteed if Byzantine failures occur.
For these reasons, a direct application of the proposed DEL method may lead to a biased calculation of the Lagrange multiplier. Consequently, the distributed empirical log-likelihood ratio statistic may not be asymptotically standard chi-squared.
Therefore, it is desirable to develop DEL in the presence of Byzantine failures. To this end, we first identify the machines without Byzantine failures and then apply the proposed DEL method to the machines without Byzantine failures.

In what follows, let us identify the Byzantine machines first by analyzing the local gradient since the gradient is transmitted for the distribution system.
For $i=1,\ldots,K$, by Taylor's expansion of $\nabla g_i(\bm{\lambda};\bm\mu_0)$ at $\bm\lambda = \bm 0$, we have
$$
\nabla g_i(\bm{\lambda};\bm\mu_0) = -\frac{1}{n}\sum_{j=1}^{n} (\bm X_{i,j} -\bm{\mu}_0) + \frac{1}{n} \sum_{j=1}^{n} (\bm X_{i,j} - \bm{\mu}_0) (\bm X_{i,j} - \bm{\mu}_0)^\top \bm{\lambda} + \bm R_i'(\bm\lambda;\bm\mu_0),
$$
where $\bm R_i'(\bm\lambda;\bm\mu_0) = \sum_{k=2}^{\infty} \nabla^k g_i(\bm 0;\bm\mu_0) \bm\lambda^{\otimes k}/k!$ is the higher-order reminder.
From Taylor's expansion, one of the main reasons for Byzantine failures may be that the sample mean in a machine does not converge to $\bm\mu_0$ because of data crashes or heterogeneity (In some practical problems, for example, different machines have different data sources for problems of multiple data sources).
Other case that may lead to Byzantine failures is that some machines may transfer wrong gradient directly.
This motivates us to consider the norm of the difference between gradients from any two worker machines as the measure for detecting the Byzantine machines.
We first consider Byzantine failures for the former case, and then extend the idea to the latter case.

For the case that the sample mean in Byzantine machines do not converge to $\bm\mu_0$, let $\mathcal S^*$ denote the index set of the worker machines without Byzantine failures.
For any fixed $i\in \mathcal S^*$, it is easy to see $\sqrt{n} \|\nabla g_i(\bm{\lambda};\bm\mu_0) \|_2 =O_p(1)$ with $\|\bm\lambda\|_2=O_p(n^{-1/2})$.
For any fixed $i \notin \mathcal{S}^*$, we have $\operatorname{E}(\bm X_{i,1})\neq \bm \mu_0$ and thus $\sqrt{n} \|\nabla g_i(\bm{\lambda};\bm\mu_0)\|_2 \to \infty$ as $n\to \infty$.
The properties of $\| \nabla g_i(\bm\lambda;\bm\mu_0)  - \nabla g_{i'}(\bm\lambda;\bm\mu_0) \|_2$ are given in Lemma \ref{lemma:norm-diff} for $i, i' \in\{1,\ldots,K\}$ and $i\neq i'$.

\begin{lemma}\label{lemma:norm-diff}
Assume the eigenvalues of $\operatorname{Cov}(\bm X_{i,1})$ are uniformly bounded away from zero and infinity for $i=1,\ldots,K$, and $K=O(N^{1-2/\beta})$ for $\beta \ge 2$.
Given $\widehat{\bm\lambda}_0$ satisfying $\| \widehat{\bm\lambda}_0 \|_2 = O_p(n^{-1/2})$, we have

(a) $\max_{i, i'\in \mathcal{S}^*, i\neq i'}\| \nabla g_i(\widehat{\bm\lambda}_0;\bm\mu_0)  - \nabla g_{i'}(\widehat{\bm\lambda}_0;\bm\mu_0)\|_2 = O_p(n^{-1/2}\sqrt{\log K})$;

(b) $\Pr \left( \min_{i\in\mathcal{S}^*,i'\notin \mathcal{S}^*} \| \nabla g_i(\widehat{\bm\lambda}_0;\bm\mu_0)  - \nabla g_{i'}(\widehat{\bm\lambda}_0;\bm\mu_0)\|_2 \ge C_n \right) \to 1$, as $n\to\infty$, for some finite constant $C_n>0$ satisfying $(\log K)^{-1/2} \sqrt{n} \leftsecond \min_{i'\notin \mathcal{S} ^*} \| \operatorname{E}(\bm X_{i',1}) - \bm\mu_0 \|_2 - C_n \rightsecond \to \infty$ as $n \to \infty$.
\end{lemma}


We assume $|\mathcal{S}^*|>K/2$ in the machine selection procedure, which is the same as the majority rule widely adopted in the invalid instrument literature \citep{kang2016instrumental,windmeijer2019use}.
Motivated by Lemma \ref{lemma:norm-diff}, we develop a machine selection algorithm to select the index set $\mathcal{S}^*$.
For $i=1,\ldots,K$, the $i$-th worker machine sends the local gradient $\nabla g_i(\widehat{\bm\lambda}_0;\widetilde{\bm\mu})$ to the central machine, where $\widehat{\bm\lambda}_0$ is an initial estimator, and $\widetilde{\bm\mu}$ is a consistent estimator for $\bm\mu$.
Since the machines without Byzantine failures are unknown, we set $\widehat{\bm\lambda}_0=\bm 0$ and $\widetilde{\bm\mu} = \arg\min_{\bm\mu} \sum_{i=1}^K \| \bm\mu - \bar{\bm\mu}_i \|_2$, where $\bar{\bm\mu}_i = n^{-1} \sum_{j=1}^n \bm X_{i,j}$ is the sample mean of the data stored in the $i$-th machine.
Let $s_i:= \# \{i':\, \| \nabla g_{i'}(\widehat{\bm\lambda}_0;\widetilde{\bm\mu}) - \nabla g_i(\widehat{\bm\lambda}_0;\widetilde{\bm\mu})\|_2 < \gamma_n, i'= 1,\ldots,K \}$ denote the number of $\nabla g_{i'}(\widehat{\bm\lambda}_0;\widetilde{\bm\mu})$ being close to $\nabla g_i(\widehat{\bm\lambda}_0;\widetilde{\bm\mu})$ for $i=1,\ldots,K$, where $\gamma_n>0$ is a pre-specified threshold.
If $s_i> K/2$, the $i$-th machine is selected as the machine without Byzantine failures and thus the index set $\mathcal S^*$ can be selected as $\mathcal{S} := \left\{i:\,s_i> K/2,i=1,\ldots,K\right\}$.
The assumption $|\mathcal S^*|>K/2$ ensures the success of the selection procedure even though the norm of the difference between two gradients from Byzantine machines may be smaller than $\gamma_n$.
The machine selection procedure is described in Algorithm \ref{alg:Choice} and the selection consistency property is summarized in Theorem \ref{thm:cluster}.

\begin{algorithm}[htbp]
\caption{Machine selection algorithm}
\label{alg:Choice}
\begin{algorithmic}[1]
\Require
An initial estimator $\widehat{\bm\lambda}_0$, a consistent estimator $\widetilde{\bm\mu}$, the threshold $\gamma_n$.

\State For $i=1,\ldots,K$, the $i$-th worker machine computes $\nabla g_i(\widehat{\bm\lambda}_0;\widetilde{\bm\mu})$ and sends it to the central machine.

\State Compute $s_i:= \# \{i':\, \| \nabla g_{i'}(\widehat{\bm\lambda}_0;\widetilde{\bm\mu}) - \nabla g_i(\widehat{\bm\lambda}_0;\widetilde{\bm\mu}) \|_2 < \gamma_n, i' = 1,\ldots,K \}$ for $i=1,\ldots,K$ and let $\mathcal S=\left\{i:\, s_i> K/2,i=1,\ldots,K\right\}$.

\Ensure
$\mathcal S$.
\end{algorithmic}
\end{algorithm}

\begin{theorem}\label{thm:cluster}
Under the conditions of Lemma \ref{lemma:norm-diff} with $C_n$ replaced by $\gamma_n$, further assume $\gamma_n\sqrt{n}/\sqrt{\log K}\to\infty$ as $n\to \infty$. If $|\mathcal{S}^{*}|> K/2$, the selection of $\mathcal{S}$ is consistent, that is,
$$
\Pr(\mathcal{S}=\mathcal{S}^{*})\to 1,~\text{as}~ n\to\infty,
$$
where $\mathcal{S}$ is defined in Algorithm \ref{alg:Choice}.
\end{theorem}

In practice, we recommend taking a conservative threshold $\gamma_n = a n^{-1/2} \log K$ for some constant $a>0$.
Using the proposed machine selection algorithm, we can obtain the index set of worker machines without Byzantine failures, denoted by $\mathcal{S}$.
By replacing the entire data set $\{\bm X_{i,j}:\, i =1,\ldots,K,\, j=1,\ldots,n\}$ with the selected data set $\{\bm X_{i,j}:\, i \in \mathcal{S},\, j=1,\ldots,n\}$, we can apply Algorithm \ref{alg:EL} to calculating the corresponding Lagrange multiplier, denoted by $\widehat{\bm\lambda}_{\mathcal{S}T}$. In the presence of Byzantines failures, the DEL method is applied to the selected worker machines in  $\mathcal{S}$ and thus we call this method the DEL$_\mathcal{S}$ method.
Given $\mathcal{S}$ and $\widehat{\bm\lambda}_{\mathcal{S}T}$, the empirical log-likelihood ratio statistic is constructed by
\begin{eqnarray}
\label{ellRDEL}
\ell_{DEL_{\mathcal{S}}}(\bm\mu)=2\sum_{i\in\mathcal{S}}\sum_{j=1}^{n}\log\leftsecond 1+\widehat{\bm\lambda}_{\mathcal{S}T}^{\top}(\bm X_{i,j}-\bm\mu)\rightsecond.
\end{eqnarray}
The asymptotic distribution of $\ell_{DEL_{\mathcal{S}}}(\bm\mu_0)$ is given in Theorem \ref{thm:RDEL}.

\begin{theorem}\label{thm:RDEL}
Under the conditions of Theorem \ref{thm:cluster}, further assume $\operatorname{E}\| \bm X_{i,1}\|_2^\beta<\infty$ for $i\in \mathcal{S}^*$, $K=O(N^{1-2/\beta})$ for $\beta \ge 4$, and $T \geq \lfloor \log K/\log n \rfloor + 1$, we have

(a) $\|\widehat{\bm\lambda}_{\mathcal{S}T} - \widehat{\bm\lambda}_{\mathcal{S}^*}\|_2=o_p(N^{-1/2}),$
where $\widehat{\bm\lambda}_{\mathcal{S}^*} = \arg\min_{\bm\lambda} \leftsecond |\mathcal{S}^*|^{-1} \sum_{i\in\mathcal{S}^*} g_i(\bm\lambda) \rightsecond$ is the Lagrange multiplier calculated using the data stored in the worker machines without Byzantine failures;

(b) $\ell_{DEL_{\mathcal{S}}}(\bm\mu_0) \stackrel{d}\rightarrow \chi_{(d)}^2$ as $n\to\infty$, where $\ell_{DEL_{\mathcal{S}}}(\bm\mu_0)$ is defined by \eqref{ellRDEL} with $\bm\mu=\bm \mu_0$.
\end{theorem}

We have considered the case that the sample means in Byzantine machines do not converge to $\bm\mu_0$.
This is a special case, which can be caused by either heterogeneity due to multiple data sources or data crash.
Another case is that some wrong gradients in some working machines are transmitted to the central machine directly. For the case, let $\bm h_i$ denote the gradient received by central machine from $i$-th worker machine, and $\nabla g_i(\bm\lambda)$ the gradient calculated by the data stored in the $i$-th machine, for $i=1,\ldots,K$.
Motivated by the analysis of the case that Byzantine machines have wrong sample means, we select the index set as $\widetilde{\mathcal{S}} = \{i:\, \widetilde{s}_i >K/2, i=1,\ldots,K \}$, where $\widetilde{s}_i = \#\{ i':\, \|\bm h_i - \bm h_{i'} \|_2 < \gamma_n, i'= 1,\ldots,K \}$ for $i=1,\ldots,K$.
We denote the index set of the worker machines without Byzantine failures as $\widetilde{\mathcal{S}}^* = \{i:\, (\log K)^{-1/2}\sqrt{n} \| \bm h_i - \nabla g_i(\bm\lambda) \|_2 <C ,i=1,\ldots,K \}$ for some constant $C>0$.
The reason is that a wrong transmission can be considered as a correct one as long as the error is small enough and it does not affect the asymptotic properties.
In this case, Theorem \ref{thm:cluster} and Theorem \ref{thm:RDEL} still hold with $\mathcal{S}^*$, $\mathcal{S}$ replaced by $\widetilde{\mathcal{S}}^*$, $\widetilde{\mathcal{S}}$, and the condition $(\log K)^{-1/2} \sqrt{n} \leftsecond \min_{i'\notin \mathcal{S} ^*} \| \operatorname{E}(\bm X_{i',1}) - \bm\mu_0 \|_2 - \gamma_n \rightsecond \to \infty$ as $n \to \infty$ replaced by $(\log K)^{-1/2} \sqrt{n} \leftsecond \min_{i'\notin \widetilde{\mathcal{S}}^*} \| \bm h_i - \nabla g_i(\bm\lambda) \|_2 - \gamma_n \rightsecond \to \infty$ as $n \to \infty$, respectively.
The reason that we consider the mean Byzantine failure first is that it is an easy and important case and provides 
help for obtaining the algorithm and theory of the general case.

\section{Simulation Studies} \label{simulations}

In this section, we conducted numerical simulations to evaluate the finite-sample performance of the DEL method and the DEL$_{\mathcal{S}}$ method, respectively.
For comparison, we took a sample size of $N=200,000$ such that the EL based on the entire data set can be calculated in the central machine.
In each simulation setting, observations $\{\bm X_{i,j}:\,i=1,\ldots,K,\,j=1,\ldots,n \}$ were generated independently from the $d$-dimensional distribution with the mean vector $\bm\mu_0$ and $d=15$.
We considered different data generation mechanisms in simulations and changed the mean of observations only in machines with Byzantine failures.
The simulation was repeated 2,000 times for the first three cases.

In the first set of simulations, we reported the Type I error rates of EL, DEL, DEL$_\mathcal{S}$ for the hypothesis test ``H$_0$: $\bm\mu=\bm\mu_0$ $\leftrightarrow$ H$_1$: $\bm\mu\neq\bm\mu_0$" at the level of significance $\alpha=0.05$.
We considered the two cases where observations $\{\bm X_{i,j}:\, i=1,\ldots,K,\, j=1,\ldots,n\}$ were generated independently from two kinds of $d$-dimensional distributions, respectively: (a) the multivariate normal distribution $\mathcal N(\bm 0,\mathbf \Sigma_1)$, where $\mathbf \Sigma_1 = (\sigma_{a,b})_{d\times d}$ with $\sigma_{a,a}=1$ for $a=1,\ldots,d$ and $\sigma_{a,b} = 0.5$ for $a\neq b$; (b) the independent copies of the joint exponential distribution $\mathbf\Sigma_2\bm E$, where $\mathbf\Sigma_2 = (\sigma_{a,b})_{d\times d}$, $\sigma_{a,b}=0.5^{|a-b|}$ for $a,b=1,\ldots,d$ and $\bm E = (e_1,\ldots,e_d)^\top$ with $e_a$ following the exponential distribution with rate $1$ independently for $a=1,\ldots,d$.
We set $\bm \mu_0=\bm 0$ for the case of the multivariate normal distribution and $\bm\mu_0 = \mathbf\Sigma_2 \cdot \bm 1$ with $\bm 1=(1,\ldots,1)^\top$ for the case of the joint exponential distribution.
We considered different values of the number of worker machines by setting $K=250,400$.
The observations in Byzantine machines were generated by changing the mean.
For the case of multivariate normal distribution, the observations in Byzantine machines were generated from $\mathcal{N}(\bm\mu_w,\mathbf\Sigma_1)$ with $\bm\mu_w = (\delta_B,\ldots,\delta_B)^\top$.
For the case of joint exponential distribution, the observations in Byzantine machines were generated as the independent copies of $\mathbf\Sigma_2 \bm E_w$ with $\bm E_w = (e_{w,1},\ldots,e_{w,d})$, where $e_{w,a}$ follows the exponential distribution with rate $1+\delta_B$ independently for $a=1,\ldots,d$.
We considered Byzantine failures by setting $\delta_B = 0.3,0.5$, respectively.
The number of Byzantine machines, denoted by $n_B$, was set to be $n_B = 0,2,10,50$.

Table \ref{table:TypeI} summarizes the Type I error rates of EL, DEL, and DEL$_\mathcal{S}$.
As shown in Table \ref{table:TypeI}, in the absence of Byzantine failures ($n_B=0$), the Type I error rates of EL, DEL, and DEL$_\mathcal{S}$ are the same in all scenarios in the sense of being accurate to the four decimal places.
In the presence of Byzantine failures ($n_B>0$), the Type I error rate of DEL increases to 1 as the number of Byzantine machines $n_B$ increases because DEL method cannot address Byzantine failures.
As expected, the Type I error rate of DEL$_\mathcal{S}$ can approximate the significance level well.

\begin{table}[htbp]
\centering
\caption{The Type I error rates of EL, DEL, DEL$_\mathcal{S}$ with $\alpha=0.05$}\label{table:TypeI}
\begin{tabular}{cccccccccc}
\hline
                     &       &  & \multicolumn{3}{c}{$\mathcal N$}             &  & \multicolumn{3}{c}{Exp}             \\ \cline{4-6} \cline{8-10}
$\delta_B$           & $n_B$ &  & EL     & DEL    & DEL$_\mathcal{S}$ &  & EL     & DEL    & DEL$_\mathcal{S}$ \\ \hline
\multicolumn{10}{c}{K=250}                                                                                     \\ \hline
\multirow{4}{*}{0.3} & $~$0  &  & 0.0510 & 0.0510 & 0.0510            &  & 0.0510 & 0.0510 & 0.0510            \\
                     & $~$2  &  & -      & 0.1170 & 0.0530            &  & -      & 0.7605 & 0.0515            \\
                     & 10    &  & -      & 1.0000 & 0.0535            &  & -      & 1.0000 & 0.0525            \\
                     & 50    &  & -      & 1.0000 & 0.0475            &  & -      & 1.0000 & 0.0510            \\ \hline
\multirow{4}{*}{0.5} & $~$0  &  & 0.0525 & 0.0525 & 0.0525            &  & 0.0520 & 0.0520 & 0.0520            \\
                     & $~$2  &  & -      & 0.2855 & 0.0520            &  & -      & 0.9995 & 0.0540            \\
                     & 10    &  & -      & 1.0000 & 0.0480            &  & -      & 1.0000 & 0.0570            \\
                     & 50    &  & -      & 1.0000 & 0.0550            &  & -      & 1.0000 & 0.0460            \\ \hline
\multicolumn{10}{c}{K=400}                                                                                     \\ \hline
\multirow{4}{*}{0.3} & $~$0  &  & 0.0510 & 0.0510 & 0.0510            &  & 0.0510 & 0.0510 & 0.0510            \\
                     & $~$2  &  & -      & 0.0765 & 0.0525            &  & -      & 0.3000 & 0.0525            \\
                     & 10    &  & -      & 0.8670 & 0.0530            &  & -      & 1.0000 & 0.0545            \\
                     & 50    &  & -      & 1.0000 & 0.0515            &  & -      & 1.0000 & 0.0475            \\ \hline
\multirow{4}{*}{0.5} & $~$0  &  & 0.0525 & 0.0525 & 0.0525            &  & 0.0520 & 0.0520 & 0.0520            \\
                     & $~$2  &  & -      & 0.1245 & 0.0525            &  & -      & 0.7835 & 0.0540            \\
                     & 10    &  & -      & 1.0000 & 0.0560            &  & -      & 1.0000 & 0.0510            \\
                     & 50    &  & -      & 1.0000 & 0.0525            &  & -      & 1.0000 & 0.0480            \\ \hline
\end{tabular}
\end{table}

In the second set of simulations, we plotted the power curves of EL, DEL, and DEL$_\mathcal{S}$ for the hypothesis test ``H$_0$: $\bm\mu=\bm\mu_0$ $\leftrightarrow$ H$_1$: $\bm\mu\neq\bm\mu_0$".
The observations $\{\bm X_{i,j}:\, i=1,\ldots,K,\, j=1,\ldots,n\}$ were generated from two kinds of $d$-dimensional distributions: (a) the multivariate normal distribution $\mathcal{N} (\bm\mu_{\delta_1},\mathbf\Sigma_1)$ with $\bm\mu_{\delta_1} = (\delta_1,\ldots,\delta_1)^\top$ and $\delta_1\in \{-0.015,\ldots,\allowbreak -0.001,0,0.001,\ldots,0.015\}$; (b) the joint exponential distribution $\mathbf\Sigma_2\bm E_{\delta_2}$, where $\bm E_{\delta_2} = (e_{\delta_2,1},\ldots,e_{\delta_2,d})^\top$ and $e_{\delta_2,a}$ follows the exponential distribution with rate $1+\delta_2$ independently with $\delta_2\in \{-0.005,\ldots,-0.0005,0,0.0005,\ldots,0.005\}$ for $a=1,\ldots,d$.
We set $\bm\mu_0=\bm 0$ for the case of the multivariate normal distribution and $\bm\mu_0=\mathbf\Sigma_2\cdot \bm 1$ for the case of the joint exponential distribution.
We kept the number of worker machines being $K=250$.
The observations in Byzantine machines were generated by changing the mean.
For the case of multivariate normal distribution, the observations in Byzantine machines were generated from $\mathcal{N}(\bm\mu_w,\mathbf\Sigma_1)$ with $\bm\mu_w = (\delta_B,\ldots,\delta_B)^\top$.
For the case of joint exponential distribution, the observations in Byzantine machines were generated as the independent copies of $\mathbf\Sigma_2 \bm E_w$ with $\bm E_w = (e_{w,1},\ldots,e_{w,d})$, where $e_{w,a}$ follows the exponential distribution with rate $1+\delta_B$ independently for $a=1,\ldots,d$.
We fixed $\delta_B = 0.3$ and took the number of Byzantine machines $n_B=0,2,5,50$.

\begin{figure}[htbp]
\centering
\subfigure[$\mathcal{N}$]{\includegraphics[width=7cm]{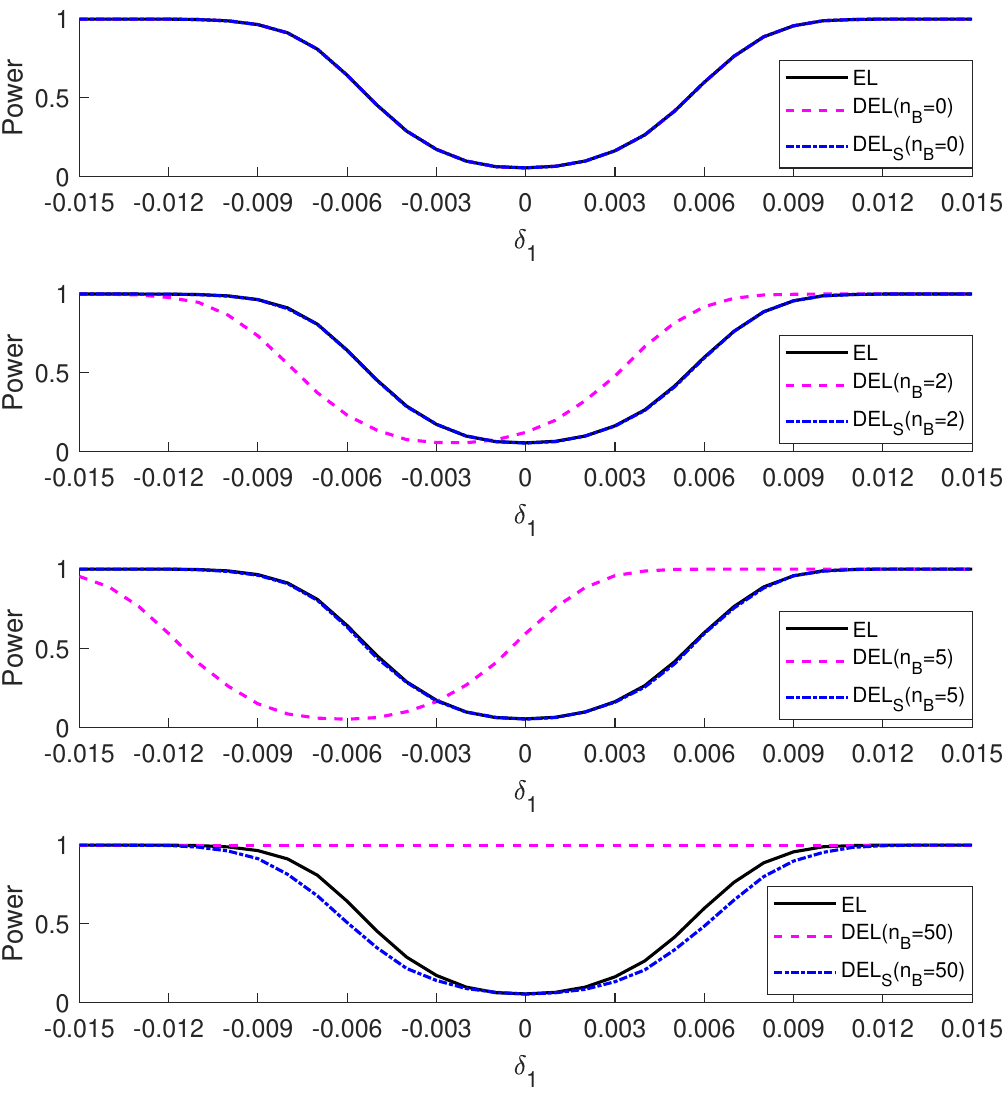}}
\subfigure[Exp]{\includegraphics[width=7cm]{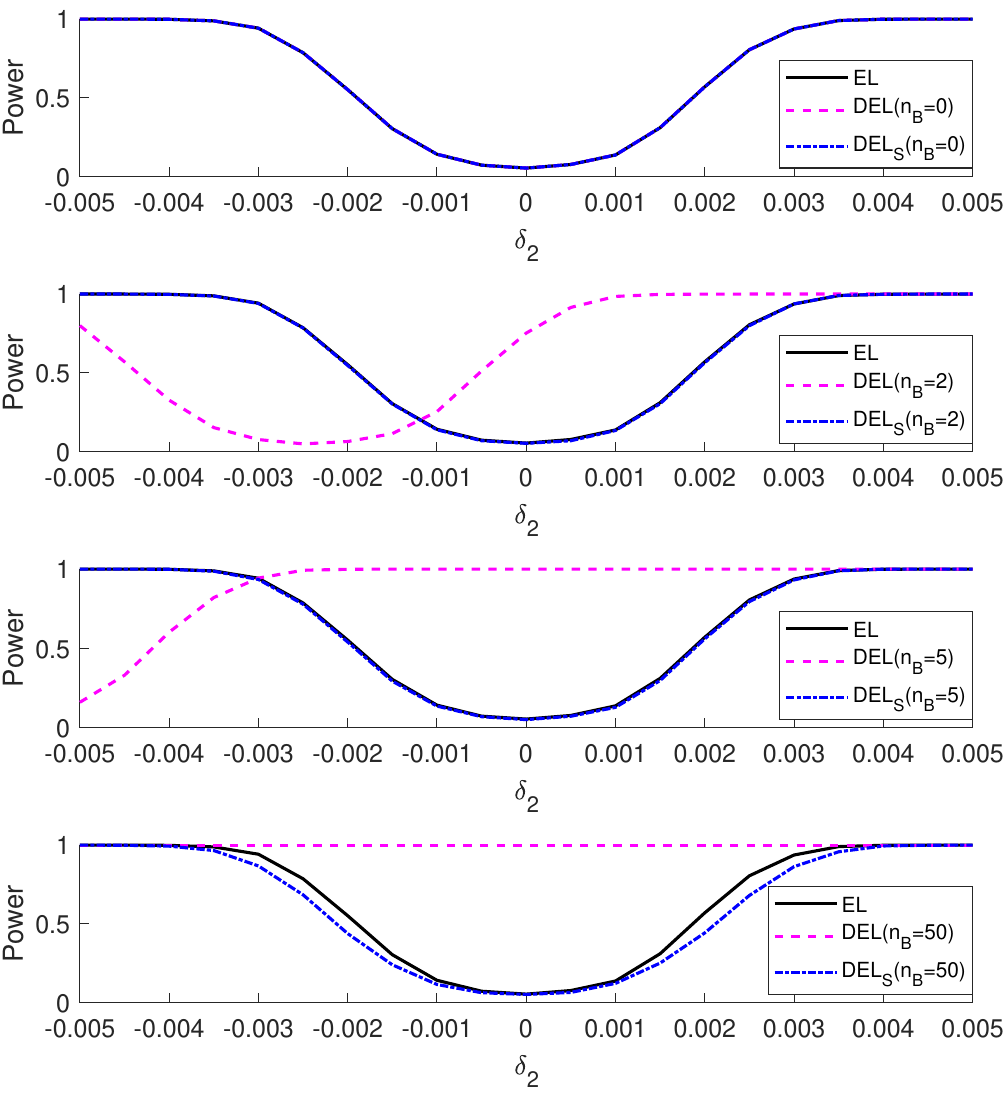}}
\caption{The power curves of EL, DEL and DEL$_\mathcal{S}$ with $n_B=0,2,5,50$, and $\alpha=0.05$: (a) the power curves under the case of multivariate normal distribution; (b) the power curves under case of joint exponential distribution.}\label{fig:Power}
\end{figure}

Figure \ref{fig:Power} plots the power curves of EL, DEL, and DEL$_\mathcal{S}$ with $\alpha=0.05$ and $n_B=0,2,5,50$ under the two data generation mechanisms.
The left and right ones are the power curves with observations generated from the multivariate normal distribution and the joint exponential distribution, respectively.
In the absence of Byzantine failures ($n_B=0$), EL, DEL, and DEL$_\mathcal{S}$ have the same power for arbitrary $\delta_1$, $\delta_2$ in our simulations under the two data generation mechanisms, which shows the validity of DEL and DEL$_\mathcal{S}$ in the absence of Byzantine failures.
In the presence of Byzantine failures ($n_B>0$), the power of DEL$_\mathcal{S}$ is slightly lower than EL's, and the power of DEL$_\mathcal{S}$ increases as $n_B$ decreases.
The reason may be that EL calculates the empirical log–likelihood ratio using the entire data set in the central machine, whereas DEL$_\mathcal{S}$ calculates the empirical log–likelihood ratio only using the samples in the worker machines without Byzantine failures.
Furthermore, in the presence of Byzantine failures ($n_B>0$), as the deviation degree of the null hypothesis increases, the power of DEL cannot increase monotonically.
An intuitive explanation is that DEL cannot eliminate the data from Byzantine machines.
Comparing with the power curve of EL, the power curve of DEL shifts left in the presence of Byzantine failures because the population mean of observations in Byzantine machines has a positive bias ($\delta_B>0$), and if we took $\delta_B<0$, the power curve of DEL would shift right.
The degree of shifting increases as $n_B$ increases.
The Type I error rate of DEL cannot approximate the significance level well even if the power of DEL is higher than the other method in some situation.
Hence, in the presence of Byzantine failures, DEL method is not reliable since its performance is heavily affected by the Byzantine failures.

In the third set of simulations, we compared the three methods in terms of the coverage probability of the confidence regions with confidence level $0.90$ in Table \ref{table:confidence}.
We took the same simulation settings as those in the simulation of the Type I error rate, including the data generation mechanisms in machines with and without Byzantine failures, $K$, $\delta_B$, and $n_B$.
As shown in Table \ref{table:confidence}, the confidence regions calculated by EL, DEL, DEL$_\mathcal{S}$ have the same coverage probability in the absence of Byzantine failures ($n_B=0$) in the sense of being accurate to the four decimal places.
In the presence of Byzantine failures ($n_B>0$), the coverage probability of DEL decreases to $0$ as the number of Byzantine machines $n_B$ increases because this method cannot address Byzantine failures.
As expected, the coverage probabilities of DEL$_\mathcal{S}$ can approximate the corresponding confidence level well.

\begin{table}[htbp]
\centering
\caption{Average coverage probabilities of EL, DEL and DEL$_\mathcal{S}$ of the confidence regions with confidence level $0.90$}\label{table:confidence}
\begin{tabular}{cccccccccc}
\hline
                     &       &  & \multicolumn{3}{c}{$\mathcal{N}$}             &  & \multicolumn{3}{c}{Exp}             \\ \cline{4-6} \cline{8-10}
$\delta_B$           & $n_B$ &  & EL     & DEL    & DEL$_\mathcal{S}$ &  & EL     & DEL    & DEL$_\mathcal{S}$ \\ \hline
\multicolumn{10}{c}{K=250}                                                                                     \\ \hline
\multirow{4}{*}{0.3} & $~$0  &  & 0.8940 & 0.8940 & 0.8940            &  & 0.9000 & 0.9000 & 0.9000            \\
                     & $~$2  &  & -      & 0.7940 & 0.8960            &  & -      & 0.1465 & 0.9020            \\
                     & 10    &  & -      & 0.0000 & 0.9010            &  & -      & 0.0000 & 0.9000            \\
                     & 50    &  & -      & 0.0000 & 0.8975            &  & -      & 0.0000 & 0.9005            \\ \hline
\multirow{4}{*}{0.5} & $~$0  &  & 0.8960 & 0.8960 & 0.8960            &  & 0.8995 & 0.8995 & 0.8995            \\
                     & $~$2  &  & -      & 0.5890 & 0.8995            &  & -      & 0.0005 & 0.8965            \\
                     & 10    &  & -      & 0.0000 & 0.8915            &  & -      & 0.0000 & 0.8975            \\
                     & 50    &  & -      & 0.0000 & 0.8920            &  & -      & 0.0000 & 0.9070            \\ \hline
\multicolumn{10}{c}{K=400}                                                                                     \\ \hline
\multirow{4}{*}{0.3} & $~$0  &  & 0.8940 & 0.8940 & 0.8940            &  & 0.9000 & 0.9000 & 0.9000            \\
                     & $~$2  &  & -      & 0.8575 & 0.8975            &  & -      & 0.5720 & 0.9060            \\
                     & 10    &  & -      & 0.0795 & 0.8890            &  & -      & 0.0000 & 0.9050            \\
                     & 50    &  & -      & 0.0000 & 0.8965            &  & -      & 0.0000 & 0.9100            \\ \hline
\multirow{4}{*}{0.5} & $~$0  &  & 0.8960 & 0.8960 & 0.8960            &  & 0.8995 & 0.8995 & 0.8995            \\
                     & $~$2  &  & -      & 0.7860 & 0.8990            &  & -      & 0.1415 & 0.8965            \\
                     & 10    &  & -      & 0.0000 & 0.8965            &  & -      & 0.0000 & 0.8940            \\
                     & 50    &  & -      & 0.0000 & 0.8865            &  & -      & 0.0000 & 0.8960            \\ \hline
\end{tabular}
\end{table}

\section{Real Data Analysis} \label{realdata}

\subsection{Airline On-time Performance Study}

To track the on-time performance of all commercial flights operated by large air carriers in USA, information about on-time, delayed, canceled, and diverted flights have been collected from October 1987 to April 2008.
The full data set contains $118,914,458$ records which is available on \url{https://community.amstat.org/jointscsg-section/dataexpo/dataexpo2009}. The delay costs are very high due to extra fuel, crew, maintenance, and so on.
Therefore, we are interested in constructing the confidence region of the average arrival delay (in minutes) and the average departure delay (in minutes), and the hypothesis testing for the average delay.
After dropping the samples whose arrival delay or departure delay is missing, we totally have $N = 116,411,531$ data points. The entire data set is randomly divided into $50$ machines.
In the following, we apply DEL and DEL$_\mathcal{S}$ to evaluating the airline on-time performance.

The stipulations from the United States Department of Transportation show that flights are on-time if they depart from the gate or arrive at the gate less than 15 minutes after their scheduled departure or arrival times.
We carried out a hypothesis test ``H$_0$: $\bm\mu \in \mathcal{A} \leftrightarrow$ H$_1$: $\bm\mu \notin \mathcal{A}$ with $\mathcal{A}:=\{\bm\mu=(\mu_1,\mu_2)^\top|\, \mu_1\le 15, \mu_2 \le 15\}$, where $\mu_1$ and $\mu_2$ denote the means of the arrival delay and departure delay, respectively.
The null hypothesis cannot be rejected at the significance level $\alpha=0.01$ both by DEL and DEL$_\mathcal{S}$.
In fact, the corresponding $p$-values are larger than $1-10^{-10}$.
Hence, we cannot consider the flights operated by large air carriers are delayed.
We also calculated the $90\%$ and $95\%$ confidence regions for the average arrival delay and the average departure delay by DEL and DEL$_\mathcal{S}$, respectively, which are displayed in Figure \ref{figure:RealCon}.
The confidence regions for the average arrival delay and the average departure delay by DEL and DEL$_\mathcal{S}$ are the same with any confidence level because none of the machines are identified as Byzantine machines.

\begin{figure}
\centering
\includegraphics[width=12cm]{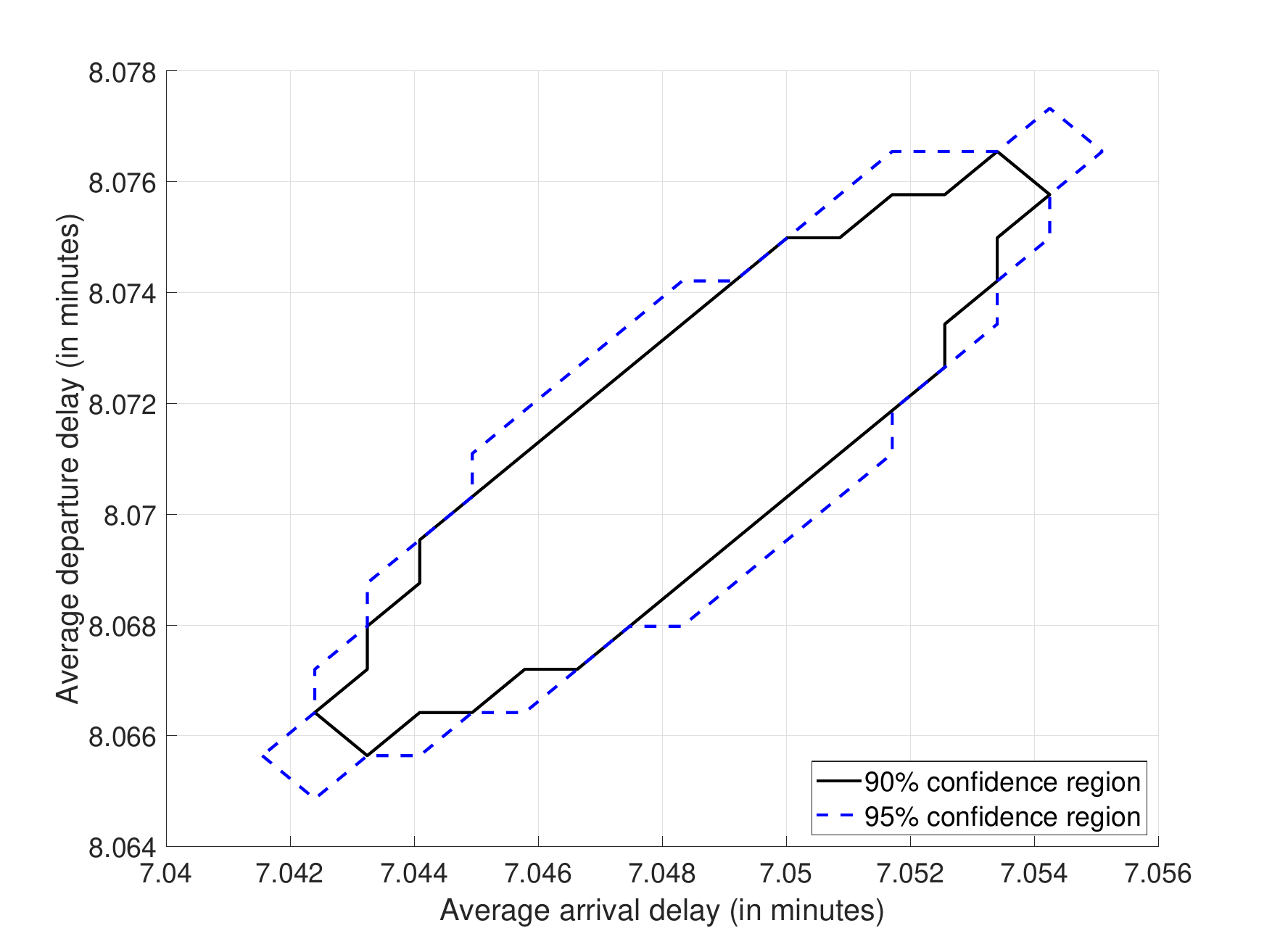}
\caption{The $90\%$ and $95\%$ confidence regions for the average arrival delay and the average departure delay calculated by DEL and DEL$_\mathcal{S}$.}\label{figure:RealCon}
\end{figure}

\subsection{Surface Climate Analysis of Yangtze River Economic Belt}

The Yangtze River Economic Belt in China is an inland river economic belt with global influence.
To analyze the main surface climate type of Yangtze River Economic Belt, we consider the Daily Value Data Set of China Surface Climate Data (V3.0), which is available on \url{http://101.200.76.197:91/mekb/?r=data/detail&dataCode=SURF_CLI_CHN_MUL_DAY_CES_V3.0}.
The data set provides daily surface climate data from  $1951$ to $2010$ among $194$ stations in China.
In our study, we focus on $83$ observation stations in the Yangtze River Economic Belt, and each station has more than $17,000$ observations.
The stations' location is plotted in Figure \ref{fig:map} (a), showing that the stations are relatively evenly distributed in the area.
The climate type is mainly determined by the precipitation and the temperature.
We are interested in constructing the confidence region of the average precipitation and the average temperature for the dominant surface climate type in the Yangtze River Economic Belt, which helps climatologist to classify the climate types.

\begin{figure}[htbp]
\centering
\subfigure[observation stations]{\includegraphics[width=7cm]{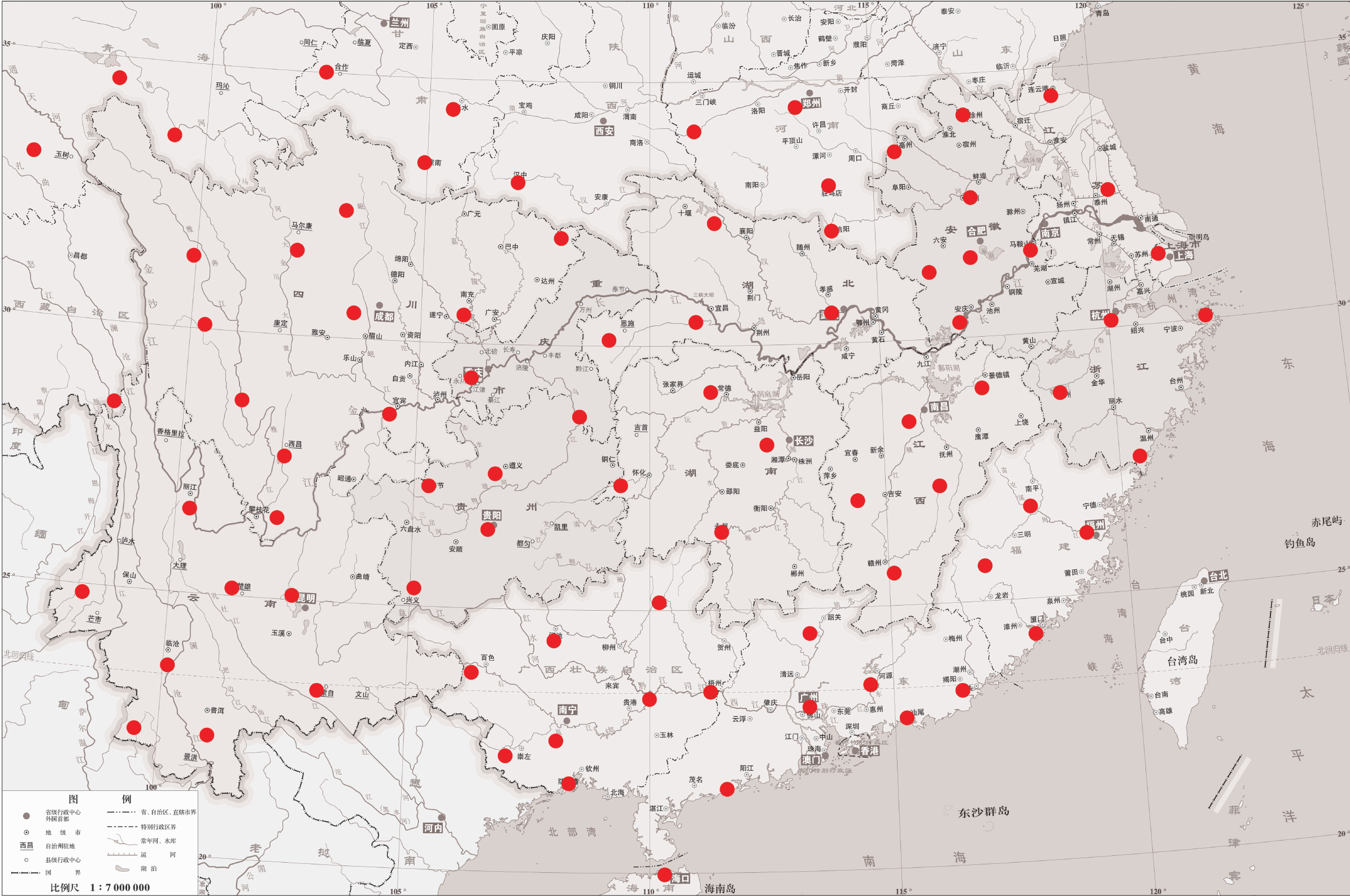}}
\subfigure[selection result]{\includegraphics[width=7cm]{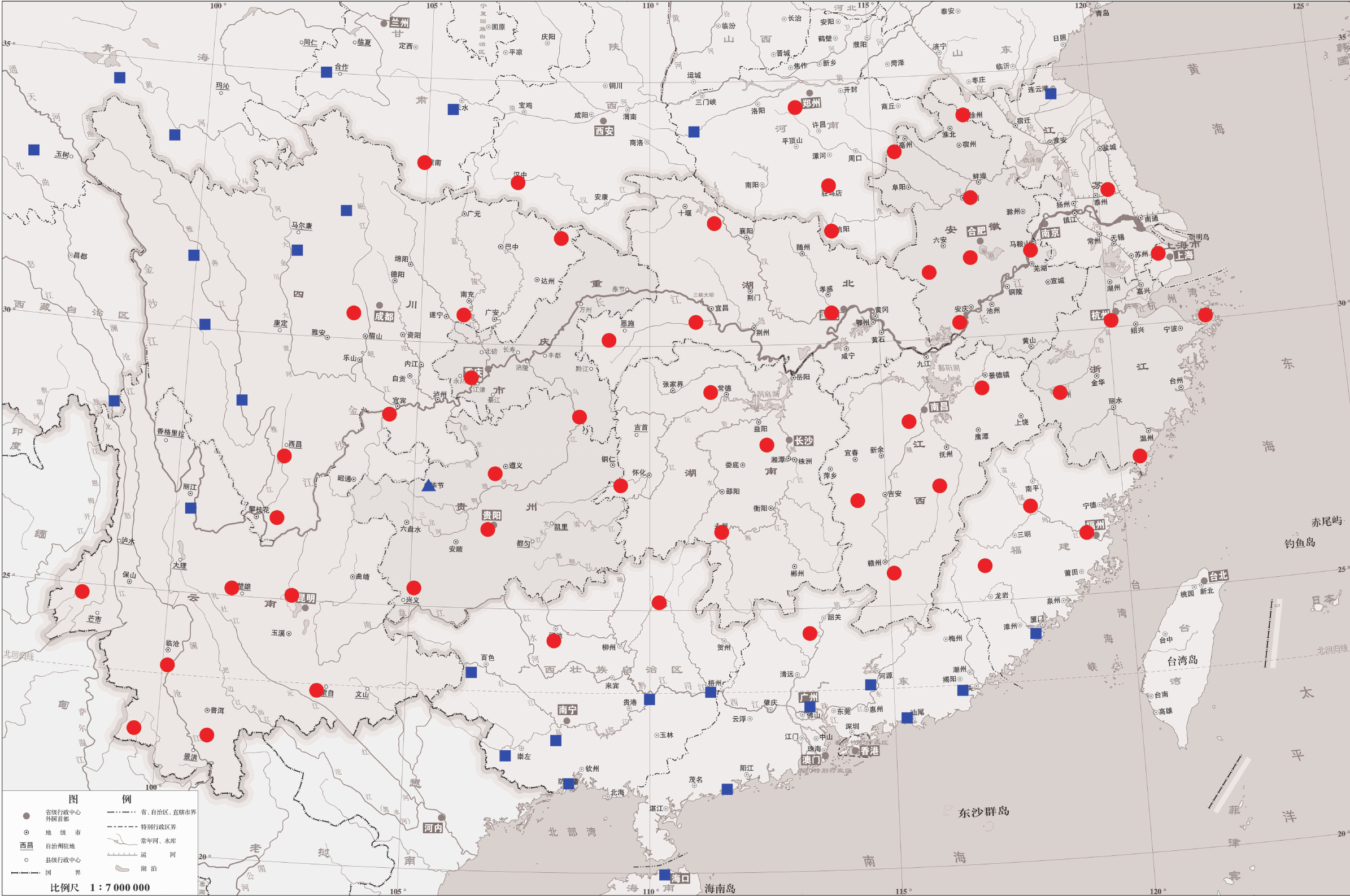}}
\caption{(a) is the location of the observation stations distributed in the Yangtze River Economic Belt. (b) is the selection result by machine selection algorithm: the red dots represent the stations with dominant surface climate type, and the blue squares and blue triangle represent the stations with other surface climate type.}\label{fig:map}
\end{figure}

The observations in different stations are regarded as local data sets stored in different worker machines.
Based on the observations of the precipitation and the temperature, we selected the stations with dominant surface climate type by Algorithm \ref{alg:Choice}, and the selection result is displayed in Figure \ref{fig:map} (b).
In Figure \ref{fig:map} (b), the red dots represent the stations with dominant surface climate type, and the blue squares and blue triangle represent the stations with other surface climate type.
For the stations in the middle of the map labelled by red dot, they have the main average precipitation and average temperature, which are in the subtropical climate zone in reality.
For the stations in the north of the map labelled by blue square, they have a lower average precipitation and average temperature, which are in the temperate climate zone and plateau mountain climate zone in reality.
For the stations in the south of the map labelled by blue square, they have a higher average precipitation and average temperature, which are in the tropical climatic zone in reality.
Specially, there exists a station labelled by the blue triangle, which is surrounded by stations labelled by red dot.
The average precipitation and the average temperature from this station labelled by the blue triangle is distinctly different from ones from the stations around it because the altitude of this station is $1,510.6$ meters, leading to a much lower average precipitation and average temperature.
Further, we calculated the $90\%$ and $95\%$ confidence regions for the average precipitation and the average temperature by DEL$_\mathcal{S}$, which are displayed in Figure \ref{figure:mapCon}.
In this case, we cannot calculate the confidence region by DEL because the homogeneity assumption for DEL is violated, and Algorithm \ref{alg:EL} does not converge due to the large difference between $\nabla g(\widehat{\bm\lambda}_t) - \nabla g_i(\widehat{\bm\lambda}_t)$ in Step 5.

\begin{figure}[htbp]
\centering
\includegraphics[width=12cm]{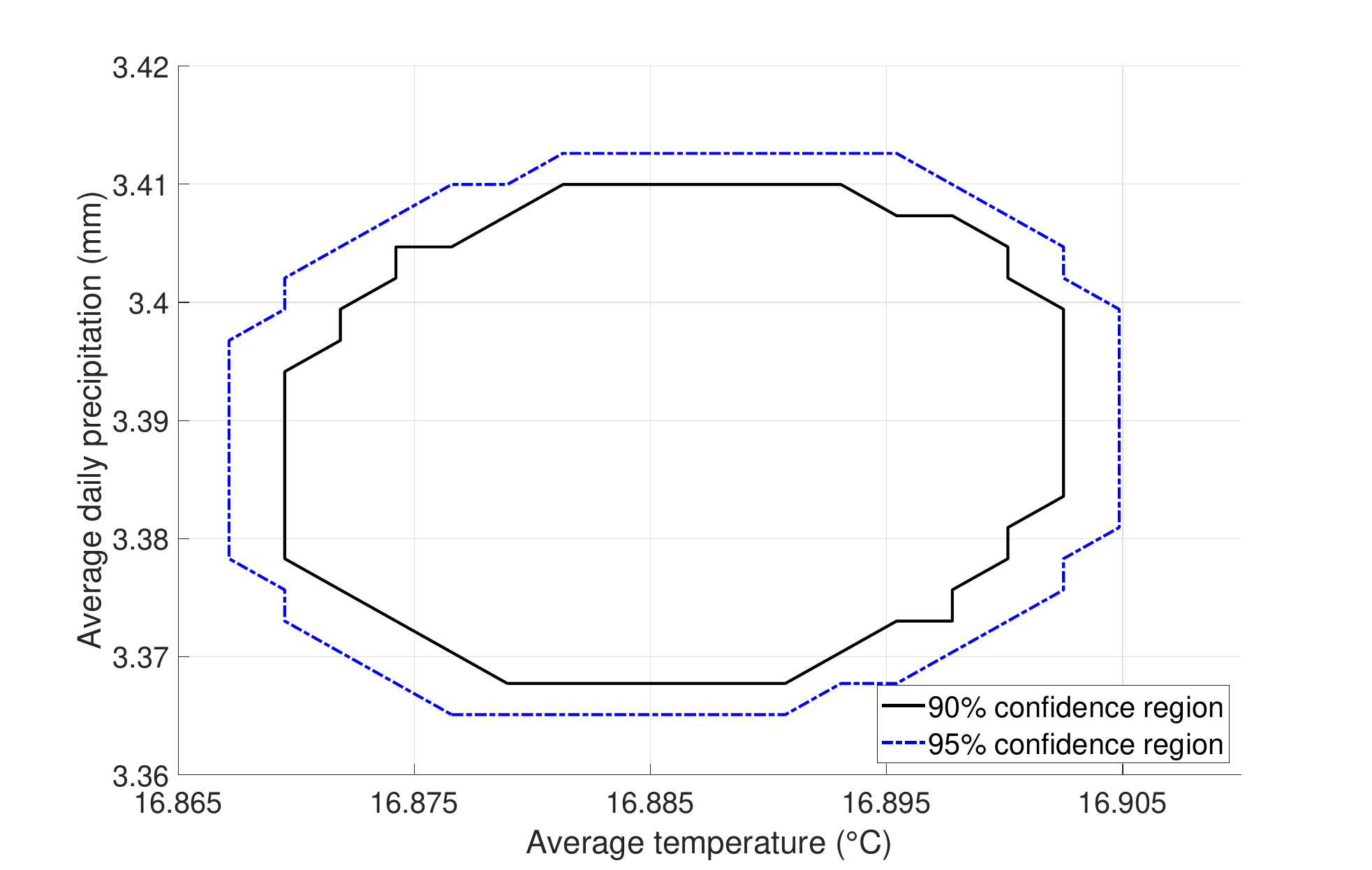}
\caption{\small The $90\%$ and $95\%$ confidence regions for the average precipitation and the average temperature calculated by DEL$_\mathcal{S}$.
}\label{figure:mapCon}
\end{figure}

\section{Discussions}

In this paper, we develop a DEL method by calculating the Lagrange multiplier in a distributed manner. By developing a machine selection algorithm, we extend DEL to address Byzantine failures. This implies that the DEL$_\mathcal{S}$ method can easily be extended to the case of heterogeneous data. The nice properties of EL introduced at the beginning of Section \ref{sec:intro} are inherited asymptotically since we have $\| \widehat{\bm\lambda}^* - \widehat{\bm\lambda}_T\|_2 = o_p(N^{-1/2})$ as long as $T$ is large enough.
The proposed methods can be easily extended to a broad range of scenarios such as parametric regression, estimation equation, and distributed jackknife empirical likelihood with nonlinear constraints, although we focus on the mean inference.


\acks{Wang’s research was supported by the National Natural Science Foundation of China (General program 11871460, General program 12271510 and program for Innovative Research Group 61621003), a grant from the Key Lab of Random Complex Structure and Data Science, CAS.
Sheng's research was supported by the National Natural Science Foundation of China (Young Scientists Fund 12201616).}



\appendix
\section{Proof of the main results}



In this appendix we provide the proofs of Lemmas \ref{Lemma:assumption} and \ref{lemma:norm-diff}, and Theorems \ref{thm:DEL}, \ref{thm:cluster}, and \ref{thm:RDEL}, and some useful lemmas.

\subsection{Preliminary Results}
The proofs require the following lemmas.

\begin{lemma}\label{lemma: order.S}
Let $\bm Z_1,\ldots,\bm Z_N$ be independent and identically distributed random vectors with $\operatorname{E}\| \bm Z_1 \|_2^\beta<\infty$, where $\|\cdot\|_2$ denotes the $\ell_2$-norm.
Let $W_N=\max_{1\leq i\leq N} \|\bm Z_i\|_2$. Then $W_N=o_p(N^{1/\beta})$.
\end{lemma}
\begin{proof}
See Proof of Lemma 11.2 in \cite{owen2001empirical}.
\end{proof}

\begin{lemma} \label{lemma: Lip.S} (Properties of convex functions)
Assume that $f: \mathbb{R}^d\to\mathbb{R}$ is convex, then $f$ is locally Lipschitz continuous on $\mathbb{R}^d$.
\end{lemma}
\begin{proof}
See Proof of Theorem 6.7 in \cite{evans2015measure}.
\end{proof}

\begin{lemma}\label{lemma:maxlambda.S}
Let $\bm X_{i,j}$, $i=1,\ldots,K$, $j=1,\ldots,n$, be independent and identically distributed random vectors with $\operatorname{E}\| \bm X_{1,1} \|_2^\beta<\infty$ with $\beta \ge 4$. Let
$$
\widehat{\bm\lambda}^* = \arg\min_{\bm\lambda} -\frac{1}{N} \sum_{i=1}^K \sum_{j=1}^n \log \leftsecond 1+\bm\lambda^\top(\bm X_{i,j} - \bm\mu)\rightsecond.
$$
Then we have
\begin{equation} \label{eq:maxlambda.S}
\max_{1\leq i \leq K,1\leq j\leq n} |\bm{\lambda}^\top(\bm X_{i,j}-\bm{\mu})| = o_p(1),
\end{equation}
for $\bm\lambda \in \mathcal B(\bm{\widehat\lambda}^*,C_1 n^{-1/2})$ and $K=O(N^{1-2/\beta})$.
\end{lemma}
\begin{proof}
It can be derived
\begin{equation}\label{lemma:S3.1.S}
\begin{aligned}
& \max_{1\leq i \leq K,1\leq j\leq n} |\bm\lambda^\top (\bm X_{i,j}-\bm{\mu})|\\
\le& \max_{1\leq i \leq K,1\leq j\leq n} | (\bm\lambda - \widehat{\bm\lambda}^*) ^\top (\bm X_{i,j}-\bm{\mu})| + \max_{1\leq i \leq K,1\leq j\leq n} |\widehat{\bm\lambda}^{*\top} (\bm X_{i,j}-\bm{\mu})|\\
\le& \max_{1\leq i \leq K,1\leq j\leq n} \| \bm\lambda - \widehat{\bm\lambda}^*\|_2 \| \bm X_{i,j}-\bm{\mu}\|_2 + \max_{1\leq i \leq K,1\leq j\leq n} \| \widehat{\bm\lambda}^*\|_2 \| \bm X_{i,j}-\bm{\mu}\|_2.
\end{aligned}
\end{equation}
It follows from Lemma \ref{lemma: order.S} that $\max_{1\leq i \leq K,1\leq j\leq n} \|\bm X_{i,j}-\bm{\mu}\|_2 = o_p(N^{1/\beta})$ under the assumption $\operatorname{E}\| \bm X_{1,1} \|_2^\beta<\infty$.
Furthermore, we have $\| \widehat{\bm\lambda}^*\|_2 = O_p(N^{-1/2})$ and $\| \bm\lambda - \widehat{\bm\lambda}^*\|_2 = O_p(n^{-1/2})$ for $\bm\lambda \in \mathcal B(\bm{\widehat\lambda}^*,C_1 n^{-1/2})$.
Hence, from $K=O(N^{1-2/\beta})$ and $n=N/K$, we have
\begin{equation}\label{lemma:S3.2.S}
\begin{aligned}
& \max_{1\leq i \leq K,1\leq j\leq n} \| \bm\lambda - \widehat{\bm\lambda}^*\|_2 \| \bm X_{i,j}-\bm{\mu}\|_2 + \max_{1\leq i \leq K,1\leq j\leq n} \| \widehat{\bm\lambda}^*\|_2 \| \bm X_{i,j}-\bm{\mu}\|_2\\
=& O_p(n^{-1/2})o_p(N^{1/\beta}) + O_p(N^{-1/2})o_p(N^{1/\beta})\\
=& O_p(K^{1/2}N^{-1/2})o_p(N^{1/\beta}) + O_p(N^{-1/2})o_p(N^{1/\beta})\\
=& o_p(1)
\end{aligned}
\end{equation}
for $\bm\lambda \in \mathcal B(\bm{\widehat\lambda}^*,C_1 n^{-1/2})$.
Combining \eqref{lemma:S3.1.S} and \eqref{lemma:S3.2.S}, we obtain \eqref{eq:maxlambda.S}.
This completes the proof of Lemma \ref{lemma:maxlambda.S}.

\end{proof}

\begin{lemma}\label{lemma:uni speed.S}
Let $\bm V_{1n},\ldots,\bm V_{Kn}$ be $d$-dimensional independent sequences of random vectors satisfying $\sqrt{n} \bm V_{in} \overset{d}{\to} \mathcal{N}(\bm 0,\mathbf A)$ as $n\to\infty$ for $i=1,\ldots,K$, where $d$ is fixed and $K$ can diverge with $n$.
Assume $c^{-1} \mathbf{I}_d \preceq \mathbf{A} \preceq c \mathbf{I}_d$ for some positive constant $c$.
The following conclusion holds
\begin{equation}\label{eq:uni speed.S}
\max_{i=1,\ldots,K} \|\bm V_{in}\|_2 = O_p(n^{-1/2} \sqrt{\log K}).
\end{equation}
\end{lemma}

\begin{proof}
(a) For the case where $K$ is fixed, \eqref{eq:uni speed.S} holds trivially because we have $\|\bm V_{in}\|_2 = O_p(n^{-1/2})$ for any fixed $i\in\{1,\ldots,K\}$.

(b) For the case where $K$ diverges as $n\to\infty$, we have
\begin{equation}\label{UnS.1.S}
\begin{aligned}
& \Pr \left( \max_{i=1,\ldots,K} \|\bm V_{in}\|_2 \le C n^{-1/2} \sqrt{\log K} \right)\\
\ge& \leftsecond \Pr \left( \|\bm V_{1n}\|_2 \le C n^{-1/2} \sqrt{\log K} \right) \rightsecond^K\\
=& \leftsecond \Pr \left( c^{-1/2} \|\bm V_{1n}\|_2 \le C c^{-1/2} n^{-1/2} \sqrt{\log K} \right) \rightsecond^K\\
\ge& \leftsecond \Pr \left( \|\mathbf{A}^{-1/2}\bm V_{1n}\|_2 \le C c^{-1/2} n^{-1/2} \sqrt{\log K} \right) \rightsecond^K,
\end{aligned}
\end{equation}
for some positive constant $C$.
The last inequation in \eqref{UnS.1.S} holds since we have $c^{-1/2}\|\bm V_{1n}\|_2 \le \|\mathbf{A}^{-1/2} \bm V_{1n}\|_2 \le c^{1/2}\|\bm V_{1n}\|_2$ under the assumption $c^{-1} \mathbf{I}_d \preceq \mathbf{A} \preceq c \mathbf{I}_d$.
It is easy to see that
\begin{equation}\label{UnS.2.S}
\begin{aligned}
& \Pr \left( \|\mathbf{A}^{-1/2}\bm V_{1n}\|_2 \le C c^{-1/2} n^{-1/2} \sqrt{\log K} \right)\\
\ge& \Pr \left( \sum_{a=1}^d \left| (\mathbf{A}^{-1/2}\bm V_{1n})_a \right| \le C c^{-1/2} n^{-1/2} \sqrt{\log K} \right)\\
\ge& \prod_{a=1}^d \Pr \left( \left| (\mathbf{A}^{-1/2}\bm V_{1n})_a \right| \le d^{-1} C c^{-1/2} n^{-1/2} \sqrt{\log K} \right),
\end{aligned}
\end{equation}
where $(\mathbf{A}^{-1/2}\bm V_{1n})_a$ denotes the $a$-th element of $\mathbf{A}^{-1/2}\bm V_{1n}$ for $a=1,\ldots,d$.
Since $\sqrt{n} \bm V_{1n} \overset{d}{\to} \mathcal{N}(\bm 0,\mathbf A)$ as $n\to \infty$, we obtain
$$
\sqrt{n} (\mathbf{A}^{-1/2}\bm V_{1n})_a \overset{d}{\to} \mathcal{N}(0,1),~\text{as}~ n\to\infty.
$$
According to Proposition 2.5 in \cite{wainwright2019high}, for a sequence of random variables $Z_n$ satisfying $\sqrt{n}Z_n \overset{d}{\to} \mathcal{N}(0,\sigma^2)$ and any $\varepsilon > 0$, the Hoeffding bound holds
\begin{equation}\label{UnS.hoe.S}
\Pr(Z_n \le \varepsilon) \ge 1 - \exp \left( -\frac{n\varepsilon^2}{2\sigma^2} \right),~\text{as}~ n\to\infty.
\end{equation}
By setting $Z_n = (\mathbf{A}^{-1/2}\bm V_{1n})_a$ and $\varepsilon = C c^{-1/2} d^{-1} n^{-1/2} \sqrt{\log K}$ in \eqref{UnS.hoe.S}, we can derive
\begin{equation}\label{UnS.3.S}
\Pr \left( \left| (\mathbf{A}^{-1/2}\bm V_{1n})_a \right| \le d^{-1} C c^{-1/2} n^{-1/2} \sqrt{\log K} \right) \ge 1 - 2 \exp \left( -\frac{C^2 \log K}{2 c d^2} \right),
\end{equation}
for $a=1,\ldots,d$.
It follows from \eqref{UnS.1.S}, \eqref{UnS.2.S} and \eqref{UnS.3.S} that
\begin{equation}\label{UnS.4.S}
\Pr \left( \max_{i=1,\ldots,K} \|\bm V_{in}\|_2 \le C n^{-1/2} \sqrt{\log K} \right) \ge \leftsecond 1 - 2 \exp \left( -\frac{C^2 \log K}{2 c d^2} \right) \rightsecond^{dK}.
\end{equation}
For some constant $C>\sqrt{2c}d$, it can be shown that
\begin{equation}\label{UnS.5.S}
\leftsecond 1 - 2 \exp \left( -\frac{C^2 \log K}{2 c d^2} \right) \rightsecond^{dK} \to 1,~\text{as}~K\to\infty.
\end{equation}
Therefore, it follows from \eqref{UnS.4.S} and \eqref{UnS.5.S} that
$$
\Pr \left( \max_{i=1,\ldots,K} \|\bm V_{in}\|_2 \le C n^{-1/2} \sqrt{\log K} \right) \to 1,~\text{as}~n\to\infty,
$$
for constant $C>\sqrt{2c}d$,
which implies \eqref{eq:uni speed.S}.
\end{proof}

\subsection{Proof of Lemma \ref{Lemma:assumption}}

(a) Strong convexity:

By the definition of $g(\bm\lambda)$, we have
$$
\nabla^2 g(\bm{\lambda})=\frac{1}{N} \sum_{i=1}^{K}\sum_{j=1}^{n} \frac{1}{\leftsecond 1+\bm{\lambda}^\top (\bm X_{i,j} - \bm{\mu})\rightsecond^2} (\bm X_{i,j}- \bm{\mu}) (\bm X_{i,j}- \bm{\mu})^\top.
$$
Using the similar arguments to the proof of Lemma \ref{lemma:maxlambda.S}, it can be shown that $$\max_{1\leq i \leq K,1\leq j\leq n} |\bm\lambda^\top (\bm X_{i,j}-\bm{\mu})| = o_p(1)$$ for $\bm{\lambda}\in \mathcal B(\bm 0,C_3 n^{-1/2})$ with a positive constant $C_3$.
This leads to
$$
\Pr \left( 0< \max_{1\leq i \leq K,1\leq j\leq n} \leftsecond 1+ \bm\lambda^\top (\bm X_{i,j} - \bm{\mu}) \rightsecond^2 \le 4 \right) \to 1,~\text{as}~n\to\infty.
$$
Therefore, it can be derived
\begin{equation}\label{PLa.1.S}
\nabla^2 g(\bm{\lambda}) \succeq \frac{1}{4}\frac{1}{N} \sum_{i=1}^{K}\sum_{j=1}^{n} (\bm X_{i,j}- \bm{\mu}) (\bm X_{i,j} - \bm{\mu})^\top,
\end{equation}
for $\bm{\lambda}\in \mathcal B(\bm 0,C_3 n^{-1/2})$ with probability tending to 1 as $n\to\infty$.
Further, we have
\begin{equation}\label{PLa.2.S}
\frac{1}{N} \sum_{i=1}^{K}\sum_{j=1}^{n} (\bm X_{i,j}- \bm{\mu}) (\bm X_{i,j} - \bm{\mu})^\top \overset{p}{\to}  \operatorname{E} (\bm X_{1,1}-\bm{\mu}) (\bm X_{1,1}-\bm{\mu})^\top, ~\text{as}~n\to \infty,
\end{equation}
where ``$\overset{p}{\to}$'' means the convergence in probability.
Under the assumption that the eigenvalues of $\operatorname{Cov}(\bm X_{1,1})$ are bounded away from zero and infinity, there exists some constant $\tau>0$ such that
\begin{equation}\label{PLa.3.S}
\operatorname{E} (\bm X_{1,1}-\bm{\mu}) (\bm X_{1,1}-\bm{\mu})^\top \succeq 4\tau \mathbf I_d.
\end{equation}
Combining \eqref{PLa.1.S}, \eqref{PLa.2.S}, and \eqref{PLa.3.S}, we have $\nabla^2 g(\bm{\lambda}) \succeq \tau \mathbf I_d$ for $\bm{\lambda}\in \mathcal B(\bm 0,C_3 n^{-1/2})$ with probability tending to 1 as $n\to\infty$, which implies the strong convexity property of $g(\bm{\lambda})$ for $\bm{\lambda}\in \mathcal B(\bm 0,C_3 n^{-1/2})$.
The strong convexity property ensures that $\widehat{\bm\lambda}^*$ is the unique minimizer of $g(\bm\lambda)$ for $\bm\lambda \in \mathcal B(\bm 0,C_3 n^{-1/2})$.
Therefore, we can define the neighborhood $\mathcal B (\widehat{\bm\lambda}^*, C_1 n^{-1/2})$ for some positive constant $C_1<C_3$.
Since $\Pr \left(\mathcal B(\bm{\widehat\lambda}^*,C_1 n^{-1/2}) \subseteq \mathcal B(\bm 0,C_3 n^{-1/2}) \right) \to 1$ as $n\to\infty$, we have $\nabla^2 g(\bm{\lambda}) \succeq \tau \mathbf I_d$ for $\bm\lambda \in \mathcal B(\widehat{\bm\lambda}^*, C_1 n^{-1/2})$ with probability tending to 1. This completes the proof of Lemma \ref{Lemma:assumption} (a).

(b) Homogeneity:

We begin the proof with the $(a,b)$-th element of the matrix $\nabla^2 g(\bm{\lambda})$.
For $a,b=1,\ldots,d$, denote by $\left( \nabla^2 g(\bm{\lambda}) \right)_{a,b}$ the $(a,b)$-th element of the matrix $\nabla^2 g(\bm{\lambda}) $ and denote by $(\bm X_{i,j}- \bm{\mu})_a$ the $a$-th element of the vector $\bm X_{i,j}- \bm{\mu}$. By some calculations, 
we have $\left( \nabla^2 g(\bm{\lambda}) \right)_{a,b} = N^{-1} \sum_{i=1}^{K}\sum_{j=1}^{n} (\mathbf{U}_{i,j})_{a,b}$ with
$$
(\mathbf{U}_{i,j})_{a,b} = \frac{(\bm X_{i,j}- \bm{\mu})_a (\bm X_{i,j}- \bm{\mu})_b}{\leftsecond 1+\bm{\lambda}^\top (\bm X_{i,j} - \bm{\mu})\rightsecond^2}.
$$
For $\bm{\lambda}\in \mathcal B(\bm{\widehat\lambda}^*,C_1 n^{-1/2})$, it can be shown that $\operatorname{Var} [ (\mathbf{U}_{i,j})_{a,b}] <\infty$.
By the definition of $\nabla^2 g(\bm\lambda)$, it is easy to see
\begin{equation}\label{PL2b.1.S}
\|\nabla^2 g(\bm\lambda) - \operatorname{E} \leftsecond \nabla^2 g(\bm\lambda) \rightsecond \|_2 = O_p(N^{-1/2}).
\end{equation}
Furthermore, for $a,b=1,\ldots,d$, we have
\begin{equation}\label{PL1.normal.S}
\sqrt{n} \leftsecond \left( \nabla^2 g_i(\bm{\lambda}) \right)_{a,b} - \operatorname{E} \left( \nabla^2 g(\bm{\lambda}) \right)_{a,b} \rightsecond \overset{d}{\to} \mathcal{N}(0,\sigma_{a,b}^2),~\text{as}~n\to\infty,
\end{equation}
with $\sigma_{a,b}^2 = \operatorname{Var} (\mathbf{U}_{1,1})_{a,b}$.
Using \eqref{PL1.normal.S} and the similar arguments to the proof of Lemma \ref{lemma:uni speed.S}, we have
\begin{equation}\label{PL2b.2.S}
\max_{i=1,\ldots,K} \| \nabla^2 g_i(\bm\lambda) - \operatorname{E} \leftsecond \nabla^2 g_i(\bm\lambda) \rightsecond \|_2 = O_p(n^{-1/2}\sqrt{\log K}).
\end{equation}
It is easy to see
\begin{equation}\label{PL2b.3.S}
\| \nabla^2 g_i(\bm\lambda) - \nabla^2 g(\bm\lambda) \|_2 \le \|\nabla^2 g(\bm\lambda) - \operatorname{E} \leftsecond \nabla^2 g(\bm\lambda) \rightsecond \|_2 + \| \nabla^2 g_i(\bm\lambda) - \operatorname{E} \leftsecond \nabla^2 g_i(\bm\lambda) \rightsecond \|_2,
\end{equation}
since $\operatorname{E} \leftsecond \nabla^2 g_i(\bm\lambda) \rightsecond = \operatorname{E} \leftsecond\nabla^2 g(\bm\lambda) \rightsecond$.
Combining \eqref{PL2b.1.S}, \eqref{PL2b.2.S} and \eqref{PL2b.3.S}, it can be known that $\max_{i=1,\ldots,K} \| \nabla^2 g_i(\bm\lambda) - \nabla^2 g(\bm\lambda) \|_2 = O_p(n^{-1/2}\sqrt{\log K})$. This completes the proof of (b).

(c) Smoothness of Hessian:

For $\bm{\lambda}\in \mathcal B(\bm{\widehat\lambda}^*,C_1 n^{-1/2})$, define
$$
f_{a,b}(\bm\lambda) = \frac{1}{N} \sum_{i=1}^{K}\sum_{j=1}^{n} \frac{|(\bm X_{i,j}- \bm{\mu})_a| |(\bm X_{i,j}- \bm{\mu})_b|}{\leftsecond 1+\bm{\lambda}^\top (\bm X_{i,j} - \bm{\mu})\rightsecond^2}.
$$
It can be shown that $f_{a,b}(\bm\lambda)$ is a convex function because
$$
\nabla^2 f_{a,b}(\bm\lambda) = \frac{6}{N} \sum_{i=1}^{K}\sum_{j=1}^{n} \frac{|(\bm X_{i,j}- \bm{\mu})_a| |(\bm X_{i,j}- \bm{\mu})_b|}{\leftsecond 1+\bm{\lambda}^\top (\bm X_{i,j} - \bm{\mu})\rightsecond^4} (\bm X_{i,j} - \bm{\mu}) (\bm X_{i,j} - \bm{\mu})^\top \succeq \bm 0
$$
holds with probability tending to $1$ as $n\to \infty$.
According to Lemma \ref{lemma: Lip.S}, for $\forall \bm{\lambda}_1,\bm{\lambda}_2 \in \mathcal B(\bm{\widehat\lambda}^*,C_1 n^{-1/2})$, there exists a uniform constant $L>0$ such that
\begin{equation}\label{eq:fab.S}
\big| f_{a,b}(\bm\lambda_1) - f_{a,b}(\bm\lambda_2) \big| \leq L \| \bm{\lambda}_1 - \bm{\lambda}_2 \|_2,
\end{equation}
for $a,b=1,\ldots,d$ with probability tending to $1$ as $n\to \infty$.

By the definition of $\infty$-norm of a matrix and the inequality \eqref{eq:fab.S}, the following conclusion holds with probability tending to $1$ as $n\to \infty$,
\begin{align*}
&\| \nabla^2 g(\bm{\lambda}_1) - \nabla^2 g(\bm{\lambda}_2) \|_\infty\\
=& \max_{1\leq a\leq d} \sum_{b=1}^{d} \left|\left( \nabla^2 g(\bm{\lambda}_1)\right)_{a,b} - \left(\nabla^2 g(\bm{\lambda}_2)\right)_{a,b}\right|\\
=& \max_{1\leq a\leq d} \sum_{b=1}^{d} \left| \frac{1}{N} \sum_{i=1}^{K}\sum_{j=1}^{n} \leftthird \frac{(\bm X_{i,j}- \bm{\mu})_a (\bm X_{i,j}- \bm{\mu})_b}{\leftsecond 1+\bm{\lambda}_1^\top (\bm X_{i,j} - \bm{\mu})\rightsecond^2} - \frac{(\bm X_{i,j}- \bm{\mu})_a (\bm X_{i,j}- \bm{\mu})_b}{\leftsecond 1+\bm{\lambda}_2^\top (\bm X_{i,j} - \bm{\mu})\rightsecond^2} \rightthird \right|\\
\le& \max_{1\leq a\leq d} \sum_{b=1}^{d} \frac{1}{N} \sum_{i=1}^{K}\sum_{j=1}^{n} \left| \frac{(\bm X_{i,j}- \bm{\mu})_a (\bm X_{i,j}- \bm{\mu})_b}{\leftsecond 1+\bm{\lambda}_1^\top (\bm X_{i,j} - \bm{\mu})\rightsecond^2} - \frac{(\bm X_{i,j}- \bm{\mu})_a (\bm X_{i,j}- \bm{\mu})_b}{\leftsecond1+\bm{\lambda}_2^\top (\bm X_{i,j} - \bm{\mu})\rightsecond^2} \right|\\
=& \max_{1\leq a\leq d} \sum_{b=1}^{d} \left| f_{a,b}(\bm\lambda_1) - f_{a,b}(\bm\lambda_2) \right|\\
\le& \max_{1\leq a\leq d} \sum_{b=1}^{d} L \| \bm{\lambda}_1 - \bm{\lambda}_2 \|_2.
\end{align*}
Since $\| \nabla^2 g(\bm{\lambda}_1) - \nabla^2 g(\bm{\lambda}_2) \|_2 \le \sqrt{d} \| \nabla^2 g(\bm{\lambda}_1) - \nabla^2 g(\bm{\lambda}_2) \|_\infty$, there exists some constant $M \ge L d \sqrt{d}$ such that $\| \nabla^2 g(\bm\lambda_1) - \nabla^2 g(\bm\lambda_2) \|_2 \leq M \|\bm\lambda_1 - \bm\lambda_2 \|_2$.
This completes the proof of (c).

\subsection{Proof of Theorem \ref{thm:DEL}}

By Lemma \ref{Lemma:assumption}, the optimal function $g(\bm{\lambda})$ satisfies the assumption of strong convexity, homogeneity, and smoothness of hessian.
By the choice of $\bm{\widehat\lambda}_0$, we know $\bm{\widehat\lambda}_0\in \mathcal B(\bm{\widehat\lambda}^*, C_1 n^{-1/2})$ with probability tending to 1.
Using similar arguments to Theorem 3.2 in \cite{fan2021communication}, for $t \ge 0$, we have
\begin{equation*}
\begin{aligned}
\|\widehat{\bm\lambda}_{t+1}-\widehat{\bm\lambda}^*\|_2 \leq & \frac{C_2 n^{-1/2}\sqrt{\log K}}{\tau_0} \| \widehat{\bm\lambda}_{t} - \widehat{\bm\lambda}^*\|_2 \times\\
&\min\left\{ 1,\frac{C_2n^{-1/2}\sqrt{\log K}}{\tau}\left( 1+\frac{M}{\tau_0} \| \widehat{\bm\lambda}_{t}-\widehat{\bm\lambda}^*\|_2 \right) \right\},
\end{aligned}
\end{equation*}
where the constants $C_2$, $M$ are defined in Lemma \ref{Lemma:assumption}, and $\tau_0 :=\sup\{ c\in[0,\tau]:\,\{ g_i(\bm{\lambda}) \}_{i=1}^{K}\, \mathrm{are}$ $c \mathrm{-strong \, convex \, in\,} \mathcal B(\bm{\widehat\lambda}^*, C_1 n^{-1/2}) \}
$ satisfies $\max\{\tau- C_2n^{-1/2}\sqrt{\log K},0\}\leq \tau_0\leq\tau$.
Therefore, we have $\| \bm{\widehat\lambda}_T - \bm{\widehat\lambda}^*\|_2 = o_p(N^{-1/2})$ if $T\geq\lfloor\log K/\log n\rfloor + 1$. This completes the proof of Theorem \ref{thm:DEL} (a).

By Taylor's expansion of $\ell(\bm{\widehat\lambda}_T;\bm{\mu}_0)$ and Theorem \ref{thm:DEL} (a), it can be proved
\begin{equation}\label{PT1.1.S}
2\sum_{i=1}^{K}\sum_{j=1}^{n}\log \leftsecond 1+ \bm{\widehat\lambda}_T^\top (\bm X_{i,j} - \bm\mu_0) \rightsecond= 2\sum_{i=1}^{K}\sum_{j=1}^{n}\log \leftsecond 1+ \bm{\widehat\lambda}^{*\top} (\bm X_{i,j} - \bm\mu_0) \rightsecond + o_p(1).
\end{equation}
By \cite{owen2001empirical}, we have
\begin{equation}\label{PT1.2.S}
2\sum_{i=1}^{K}\sum_{j=1}^{n}\log \leftsecond 1+ \bm{\widehat\lambda}^{*\top} (\bm X_{i,j} - \bm\mu_0) \rightsecond \overset{d}{\to} \chi_{(d)}^2,~\text{as}~n\to\infty,
\end{equation}
where ``$\overset{d}{\to}$'' denotes convergence in distribution.
It follows from \eqref{PT1.1.S} and \eqref{PT1.2.S} that $\ell(\widehat{\bm\lambda}_T;\bm\mu_0) \stackrel{d}\rightarrow \chi_{(d)}^2$.
This completes the proof of Theorem \ref{thm:DEL} (b).

\subsection{Proof of Lemma \ref{lemma:norm-diff}}

There are two scenarios for the initial value $\widehat{\bm\lambda}_0$. If we have information that the $i_0$-th machine is without Byzantine failures, we can take $\widehat{\bm\lambda}_0 = \arg\min_{\bm\lambda} g_{i_0}(\bm\lambda;\bm\mu_0)$. Otherwise, let $\widehat{\bm\lambda}_0 =\bm 0$. In the following, we first prove Lemma \ref{lemma:norm-diff} for $\widehat{\bm\lambda}_0 = \arg\min_{\bm\lambda} g_{i_0}(\bm\lambda;\bm\mu_0)$.

By \cite{owen2001empirical}, we have
\begin{equation}\label{eq:ini lambda.S}
\widehat{\bm\lambda}_0 = \mathbf\Sigma^{-1} \frac{1}{n} \sum_{j=1}^n (\bm X_{i_0,j}-\bm\mu_0) + o_p(n^{-1/2}),
\end{equation}
where $\mathbf\Sigma = \operatorname{E} \leftsecond(\bm X_{i,1} - \bm\mu_0) (\bm X_{i,1} - \bm\mu_0)^\top\rightsecond$ for any $i\in\mathcal{S}^*$, and hence $\widehat{\bm\lambda}_0 = O_p(n^{-1/2})$.
Let $\nabla g(\widehat{\bm\lambda}_0)$ denote $\nabla g(\widehat{\bm\lambda}_0;\bm\mu_0)$ and $\nabla g_i(\widehat{\bm\lambda}_0)$ denote $\nabla g_i(\widehat{\bm\lambda}_0;\bm\mu_0)$ for $i=1,\ldots,K$.

For $i=1,\ldots,K$, by Taylor's expansion, we have
\begin{equation}\label{SZ.Taylor.S}
\nabla g_i(\widehat{\bm\lambda}_0) = -\frac{1}{n} \sum_{j=1}^n (\bm X_{i,j} - \bm\mu_0) + \frac{1}{n} \sum_{j=1}^n (\bm X_{i,j} - \bm\mu_0) (\bm X_{i,j} - \bm\mu_0)^\top \widehat{\bm\lambda}_0 + \bm R_i'(\widehat{\bm\lambda}_0),
\end{equation}
where $\bm R_i'(\widehat{\bm\lambda}_0) = \sum_{k=2}^{\infty} \nabla^k g_i(\bm 0;\bm\mu_0) \widehat{\bm\lambda}_0^{\otimes k}/k!$ is the higher-order reminder.

(i) Proof of Lemma \ref{lemma:norm-diff} (a):

For some positive constant $C$, we have
\begin{equation}\label{SZ1.eq1.S}
\begin{aligned}
&\Pr \left(\max_{i,i'\in \mathcal S^*, i\ne i'} \|\nabla g_i(\widehat{\bm\lambda}_0) - \nabla g_{i'}(\widehat{\bm\lambda}_0) \|_2 \le C n^{-1/2} \sqrt{\log K} \right)\\
\ge& \Pr \left(\max_{i,i'\in \mathcal S^*, i\ne i'} \|\nabla g_i(\widehat{\bm\lambda}_0) - \nabla g_{i'}(\widehat{\bm\lambda}_0) \|_2 \le C n^{-1/2} \sqrt{\log |\mathcal{S}^*|} \right)\\
\ge& \Pr\left( \max_{i,i'\in \mathcal S^*, i\ne i'} \leftsecond \|\nabla g_i(\widehat{\bm\lambda}_0) - \nabla g_{i_0}(\widehat{\bm\lambda}_0) \|_2 + \|\nabla g_{i'}(\widehat{\bm\lambda}_0) - \nabla g_{i_0}(\widehat{\bm\lambda}_0) \|_2 \rightsecond\rightdot\\
&\quad\quad\leftdot\le C n^{-1/2} \sqrt{\log |\mathcal{S}^*|} \right)\\
\ge& \Pr\left( \max_{i\in\mathcal{S}^*} \|\nabla g_i(\widehat{\bm\lambda}_0) - \nabla g_{i_0}(\widehat{\bm\lambda}_0) \|_2 \le \frac{C n^{-1/2}}{2} \sqrt{\log |\mathcal{S}^*|} \right).
\end{aligned}
\end{equation}
Therefore, we prove Lemma \ref{lemma:norm-diff} (a) by studying the property of $\max_{i\in\mathcal{S}^*} \|\nabla g_i(\widehat{\bm\lambda}_0) - \nabla g_{i_0}(\widehat{\bm\lambda}_0) \|_2$.
It is noted that $n^{-1}\sum_{j=1}^n (\bm X_{i,j} - \bm\mu_0) (\bm X_{i,j} - \bm\mu_0)^\top = \mathbf\Sigma +o_p(1)$.
This together with \eqref{eq:ini lambda.S} and \eqref{SZ.Taylor.S} proves
\begin{equation}\label{SZ1.eq2.S}
\nabla g_i(\widehat{\bm\lambda}_0) = -\frac{1}{n} \sum_{j=1}^n (\bm X_{i,j} - \bm\mu_0) + \frac{1}{n} \sum_{j=1}^n (\bm X_{i_0,j} - \bm\mu_0) + o_p(n^{-1/2}),
\end{equation}
for $i\in\mathcal{S}^*$ and $i\neq i_0$.
Furthermore, according to the assumption that $\widehat{\bm\lambda}_0 = \arg\min_{\bm\lambda} g_{i_0}(\bm\lambda;\bm\mu_0)$, we have $\nabla g_{i_0}(\widehat{\bm\lambda}_0) = \bm 0$.
By \eqref{SZ1.eq2.S}, we have
\begin{equation}\label{SZ1.normal.S}
\sqrt{n} \leftsecond \nabla g_i(\widehat{\bm\lambda}_0) - \nabla g_{i_0}(\widehat{\bm\lambda}_0) \rightsecond \overset{d}{\to} \mathcal{N}(\bm 0,2\mathbf\Sigma),~\text{as}~n\to\infty
\end{equation}
for $i\in\mathcal{S}^*$ and $i\neq i_0$.
Under the conditions of Lemma \ref{lemma:norm-diff}, the eigenvalues of $2 \mathbf{\Sigma}$ are bounded away from zero and infinity.
From \eqref{SZ1.normal.S}, it can be proved
\begin{equation}\label{SZ1.eq3.S}
\max_{i\in\mathcal{S}^*} \|\nabla g_i(\bm\lambda) - \nabla g_{i_0}(\bm\lambda) \|_2 = O_p(n^{-1/2}\sqrt{\log |\mathcal{S}^*|})
\end{equation}
by using similar techniques in the proof of Lemma \ref{lemma:uni speed.S}.
By \eqref{SZ1.eq3.S}, for $\forall \varepsilon >0$, there exists some constant $M'>0$ such that
\begin{equation}\label{SZ1.eq4.S}
\Pr \left( \max_{i\in\mathcal{S}^*} \|\nabla g_i(\bm\lambda) - \nabla g_{i_0}(\bm\lambda) \|_2 \le M' n^{-1/2}\sqrt{\log |\mathcal{S}^*|} \right) \ge 1-\varepsilon.
\end{equation}
Combining \eqref{SZ1.eq1.S} and \eqref{SZ1.eq4.S}, and setting $C=2M'$, we have
$$
\Pr \left( \max_{i,i'\in \mathcal S^*, i\ne i'} \|\nabla g_i(\widehat{\bm\lambda}_0) - \nabla g_{i'}(\widehat{\bm\lambda}_0) \|_2 \le C n^{-1/2} \sqrt{\log K} \right) \ge 1- \varepsilon,
$$
which proves Lemma \ref{lemma:norm-diff} (a).

(ii) Proof of Lemma \ref{lemma:norm-diff} (b):

It can be shown that
\begin{equation*}
\begin{aligned}
& \Pr \left( \min_{i\in\mathcal{S}^*,i'\notin \mathcal{S}^*} \| \nabla g_i(\widehat{\bm\lambda}_0) - \nabla g_{i'}(\widehat{\bm\lambda}_0)\|_2 \ge C_n \right)\\
\ge& \Pr \left( \min_{i'\notin \mathcal{S}^*} \| \nabla g_{i'}(\widehat{\bm\lambda}_0) - \nabla g_{i_0}(\widehat{\bm\lambda}_0) \|_2 - \max_{i\in\mathcal{S}^*} \| \nabla g_{i}(\widehat{\bm\lambda}_0) - \nabla g_{i_0}(\widehat{\bm\lambda}_0) \|_2 \ge C_n \right)\\
=& \Pr \left( \min_{i'\notin \mathcal{S}^*} \| \nabla g_{i'}(\widehat{\bm\lambda}_0) - \nabla g_{i_0}(\widehat{\bm\lambda}_0) \|_2 \ge C_n' \right),
\end{aligned}
\end{equation*}
where $C_n':= C_n + \max_{i\in\mathcal{S}^*} \| \nabla g_{i}(\widehat{\bm\lambda}_0) - \nabla g_{i_0}(\widehat{\bm\lambda}_0) \|_2$.
Therefore, we prove Lemma \ref{lemma:norm-diff} (b) by studying the property of $\| \nabla g_{i'}(\widehat{\bm\lambda}_0) - \nabla g_{i_0}(\widehat{\bm\lambda}_0) \|_2$. It is noted $\nabla g_{i_0}(\widehat{\bm\lambda}_0) = \bm 0$ according to the definition of $\widehat{\bm\lambda}_0$.
Hence, it follows from \eqref{eq:ini lambda.S} and \eqref{SZ.Taylor.S} that
\begin{equation}\label{SZ2.eq2.S}
\nabla g_{i'}(\widehat{\bm\lambda}_0)-\nabla g_{i_0}(\widehat{\bm\lambda}_0) = - \bm \Delta_{i'} + \bm B_{i'},
\end{equation}
where $\bm \Delta_{i'} = \operatorname{E}(\bm X_{i',1}) - \bm\mu_0$ and
\begin{equation}\label{SZ2.defB.S}
\begin{aligned}
\bm B_{i'} =& - \frac{1}{n} \sum_{j=1}^n \leftsecond\bm X_{i',j} - \operatorname{E}(\bm X_{i',1}) \rightsecond  + \mathbf\Sigma_{i'} \mathbf\Sigma^{-1} \frac{1}{n} \sum_{j=1}^n (\bm X_{i_0,j} - \bm\mu_0)\\
& + \bm \Delta_{i'} \bm \Delta_{i'}^\top \mathbf\Sigma^{-1} \frac{1}{n} \sum_{j=1}^n (\bm X_{i_0,j} - \bm\mu_0) + \bm r_{i'}(\widehat{\bm\lambda}_0),
\end{aligned}
\end{equation}
with $\mathbf\Sigma_{i'} = \operatorname{E} \leftsecond \bm X_{i',1} - \operatorname{E}(\bm X_{i',1}) \rightsecond \leftsecond \bm X_{i',1} - \operatorname{E}(\bm X_{i',1}) \rightsecond^\top$ and $\|\bm r_{i'}(\widehat{\bm\lambda}_0) \|_2 = o_p(n^{-1/2})$ for $i'\notin\mathcal{S}^*$.
Let $\mathcal{J}:= \{i'\notin\mathcal{S}^*:\, \|\bm\Delta_{i'}\|_2 < \infty\}$. We can derive
\begin{equation}\label{SZ2.eq3.S}
\begin{aligned}
& \Pr \left( \min_{i'\notin \mathcal{S}^*} \| \nabla g_{i'}(\widehat{\bm\lambda}_0) - \nabla g_{i_0}(\widehat{\bm\lambda}_0) \|_2 \ge C_n' \right)\\
=& \Pr \left( \min_{i'\in \mathcal{J}} \| \nabla g_{i'}(\widehat{\bm\lambda}_0) - \nabla g_{i_0}(\widehat{\bm\lambda}_0) \|_2 \ge C_n' \right)\\
& \times \Pr \left( \min_{i'\notin (\mathcal{S}^* \cup \mathcal{J})} \| \nabla g_{i'}(\widehat{\bm\lambda}_0) - \nabla g_{i_0}(\widehat{\bm\lambda}_0) \|_2 \ge C_n' \right),
\end{aligned}
\end{equation}
since $\{ i':\, i'\notin \mathcal{S}^* \} \cup \{ i':\, i'\notin (\mathcal{S}^* \cup \mathcal{J}) \} = \{ i':\, i'\notin \mathcal{S}^* \}$ and $\{ i':\, i'\notin \mathcal{S}^* \} \cup \{ i':\, i'\notin (\mathcal{S}^* \cap \mathcal{J}) \} = \varnothing$.

For the first term on the right hand side of \eqref{SZ2.eq3.S}, it can be derived
\begin{equation}\label{SZ2.eq4.S}
\begin{aligned}
& \Pr \left( \min_{i'\in \mathcal{J}} \| \nabla g_{i'}(\widehat{\bm\lambda}_0) - \nabla g_{i_0}(\widehat{\bm\lambda}_0) \|_2 \ge C_n' \right)\\
\ge& \Pr \left( \min_{i'\in \mathcal{J}} \left( \|\bm \Delta_{i'}\|_2 - \|\bm B_{i'}\|_2 \right) \ge C_n' \right)\\
\ge& \Pr \left( \min_{i'\in \mathcal{J}} \|\bm \Delta_{i'}\|_2 - C_n \ge \max_{i'\in \mathcal{J}} \|\bm B_{i'}\|_2 + \max_{i\in\mathcal{S}^*} \| \nabla g_{i}(\widehat{\bm\lambda}_0) - \nabla g_{i_0}(\widehat{\bm\lambda}_0) \|_2 \right).
\end{aligned}
\end{equation}
According to \eqref{SZ2.defB.S}, by Central Limit Theorem, we have
\begin{equation}\label{SZ2.normal.S}
\begin{aligned}
\sqrt{n} \bm B_{i'} \overset{d}{\to} \mathcal{N}(\bm 0, \mathbf{\Psi_{i'}}),~\text{as}~n\to\infty,
\end{aligned}
\end{equation}
for $i'\in\mathcal{J}$, where
\begin{equation*}
\mathbf\Psi_{i'} = \left\{
\begin{aligned}
&\mathbf\Sigma_{i'} + \mathbf\Sigma_{i'} \mathbf\Sigma^{-1} \mathbf\Sigma_{i'}, \quad &\text{if}~\|\bm\Delta_{i'}\|_2 = o(1),\\
&\mathbf\Sigma_{i'} + \left( \mathbf\Sigma_{i'} + \bm \Delta_{i'} \bm \Delta_{i'}^\top \right) \mathbf\Sigma^{-1} \left( \mathbf\Sigma_{i'} + \bm \Delta_{i'} \bm \Delta_{i'}^\top \right), \quad &\text{otherwise}.
\end{aligned}
\right.
\end{equation*}
It is easy to see that the eigenvalues of $\mathbf{\Psi}_{i'}$ are bounded away from zero and infinity under the conditions of Lemma \ref{lemma:norm-diff}.
Therefore, Lemma \ref{lemma:uni speed.S} together with \eqref{SZ2.normal.S} proves
\begin{equation}\label{SZ2.eq5.S}
\max_{i'\in \mathcal{J}} \|\bm B_{i'}\|_2 = O_p(n^{-1/2}\sqrt{\log K}).
\end{equation}
By Lemma \ref{lemma:norm-diff} (a) and \eqref{SZ2.eq5.S}, it can be derived
\begin{equation}\label{SZ2.eq6.S}
(\log K)^{-1/2} \sqrt{n} \leftsecond \max_{i'\in \mathcal{J}} \|\bm B_{i'}\|_2 + \max_{i\in\mathcal{S}^*} \| \nabla g_{i}(\widehat{\bm\lambda}_0) - \nabla g_{i_0}(\widehat{\bm\lambda}_0) \|_2 \rightsecond = O_p(1).
\end{equation}
According to \eqref{SZ2.eq6.S} and the assumption $(\log K)^{-1/2} \sqrt{n} \{ \min_{i'\notin\mathcal{S}^*}\| \bm \Delta_{i'} \|_2 - C_n \} \to \infty$ as $n\to\infty$, we have
\begin{equation}\label{SZ2.eq7.S}
\Pr \left( \min_{i'\in \mathcal{J}} \|\bm \Delta_{i'}\|_2 - C_n \ge \max_{i'\in \mathcal{J}} \|\bm B_{i'}\|_2 + \max_{i\in\mathcal{S}^*} \| \nabla g_{i}(\widehat{\bm\lambda}_0) - \nabla g_{i_0}(\widehat{\bm\lambda}_0) \|_2 \right) \to 1,
\end{equation}
as $n\to\infty$.
Combining \eqref{SZ2.eq4.S} and \eqref{SZ2.eq7.S}, we have
\begin{equation}\label{SZ2.eq8.S}
\Pr \left( \min_{i'\in \mathcal{J}} \| \nabla g_{i'}(\widehat{\bm\lambda}_0) - \nabla g_{i_0}(\widehat{\bm\lambda}_0) \|_2 \ge C_n' \right)
\to 1,~\text{as}~n\to\infty.
\end{equation}

According to \eqref{SZ2.eq2.S}, for $i' \notin (\mathcal{S}^* \cup \mathcal{J})$, the dominate term of $\|\nabla g_{i'}(\widehat{\bm\lambda}_0)\|_2$ is $\| - \bm\Delta_{i'} + \bm \Delta_{i'} \bm \Delta_{i'}^\top \mathbf\Sigma^{-1} n^{-1} \sum_{j=1}^n (\bm X_{i_0,j} - \bm\mu_0) \|_2$.
For the second term on the right hand side of \eqref{SZ2.eq3.S}, we then have
\begin{equation}\label{SZ2.eq9.S}
\begin{aligned}
& \Pr \left( \min_{i'\notin (\mathcal{S}^* \cup \mathcal{J})} \| \nabla g_{i'}(\widehat{\bm\lambda}_0) - \nabla g_{i_0}(\widehat{\bm\lambda}_0) \|_2 \ge C_n' \right)\\
\ge& \Pr \left( \min_{i'\notin (\mathcal{S}^* \cup \mathcal{J})} \| - \bm\Delta_{i'} + \bm \Delta_{i'} \bm \Delta_{i'}^\top \mathbf\Sigma^{-1} n^{-1} \sum_{j=1}^n (\bm X_{i_0,j} - \bm\mu_0) \|_2 \to \infty \right)\\
=& \Pr \left( \min_{i\in\mathcal{D}_1} \| \bm \Delta_{i'} \bm \Delta_{i'}^\top \mathbf\Sigma^{-1} n^{-1} \sum_{j=1}^n (\bm X_{i_0,j} - \bm\mu_0) - \bm\Delta_{i'} \|_2 \to \infty \right) \\
& \times \Pr \left( \min_{i\in\mathcal{D}_2} \| \bm \Delta_{i'} \bm \Delta_{i'}^\top \mathbf\Sigma^{-1} n^{-1} \sum_{j=1}^n (\bm X_{i_0,j} - \bm\mu_0) - \bm\Delta_{i'} \|_2 \to \infty \right),
\end{aligned}
\end{equation}
where $\mathcal D_1=\{i':\, i'\notin (\mathcal{S}^* \cup \mathcal{J}), \|\bm\Delta_{i'}\|_2 /\sqrt{n} ~\text{is bounded away from zero and infinity}\}$ and $\mathcal D_2=\{i':\, i'\notin (\mathcal{S}^* \cup \mathcal{J}), \|\bm\Delta_{i'}\|_2 /\sqrt{n} \to 0 ~\text{or}~\infty~\text{as}~n\to \infty\}$ are two disjoint index sets.

For $i'\in \mathcal{D}_1$, let $\mathbf{\Pi}_{i'} = n^{-1} \bm\Delta_{i'} \bm\Delta_{i'}^\top$ and we have $\left\| \mathbf{\Pi}_{i'} \mathbf{\Sigma}^{-1/2} \right\|_2$ is bounded away from zero and infinity since $\|\bm\Delta_{i'}\|_2/\sqrt{n}$ is bounded away from zero and infinity. Let $\bm Z=\mathbf{\Sigma}^{-1/2} n^{-1/2} \sum_{j=1}^n (\bm X_{i_0,j} - \bm\mu_0)$ and we have $\bm Z \overset{d}{\to} \mathcal{N}(\bm 0, \mathbf{I}_d)$ as $n\to \infty$ by Central Limit Theorem.
Therefore, we have
\begin{eqnarray}
\label{D1}
&&\Pr \left( \min_{i'\in \mathcal{D}_1} 
\| \bm\Delta_{i'} \bm\Delta_{i'}^\top \mathbf{\Sigma}^{-1} n^{-1} \sum_{j=1}^n (\bm X_{i_0,j} - \bm\mu_0) - \bm\Delta_{i'} \|_2 \to \infty \right)\\\nonumber
&=&
\Pr \left( \min_{i'\in \mathcal{D}_1} \sqrt{n} \left\| \mathbf{\Pi}_{i'} \mathbf{\Sigma}^{-1/2} \bm Z - n^{-1/2} \bm\Delta_{i'} \right\|_2 \to \infty \right) \\\nonumber
&\to& 1, ~\text{as}~ n\to \infty,
\end{eqnarray}
by using similar techniques in the proof of Lemma \ref{lemma:uni speed.S}.

Moreover, it is easy to see 
\begin{eqnarray}
\label{D2}
\Pr \left( \min_{i'\in\mathcal{D}_2} \| \bm \Delta_{i'} \bm \Delta_{i'}^\top \mathbf\Sigma^{-1} n^{-1} \sum_{j=1}^n (\bm X_{i_0,j} - \bm\mu_0) - \bm\Delta_{i'} \|_2 \to \infty \right) \to 1,~\text{as}~ n \to \infty.
\end{eqnarray} 

By \eqref{SZ2.eq9.S}, \eqref{D1}, and \eqref{D2}, we have
\begin{equation}\label{SZ2.eq10.S}
\Pr \left( \min_{i'\notin (\mathcal{S}^* \cup \mathcal{J})}  \| \nabla g_{i'}(\widehat{\bm\lambda}_0) - \nabla g_{i_0}(\widehat{\bm\lambda}_0) \|_2 \ge C_n' \right)\to 1,~\text{as}~ n \to \infty. 
\end{equation}
Combining \eqref{SZ2.eq3.S}, \eqref{SZ2.eq8.S}, and \eqref{SZ2.eq10.S}, we then prove Lemma \ref{lemma:norm-diff} (b).

For the case where we have no information about the machine without Byzantine failures, we can take $\widehat{\bm\lambda}_0 = \bm 0$ as the initial estimator.
Then $\nabla g_i(\widehat{\bm\lambda}_0) - \nabla g_{i'}(\widehat{\bm\lambda}_0) = -n^{-1} \sum_{j=1}^n (\bm X_{i,j} - \bm X_{i',j})$.
In this situation, we can complete the proof by using similar techniques to the proof with $\widehat{\bm\lambda}_0 = \arg\min_{\bm\lambda} g_{i_0}(\bm\lambda;\bm\mu_0)$.
In the following, we present some key steps for the proof with $\widehat{\bm\lambda}_0 = \bm 0$.

(i') Proof of Lemma \ref{lemma:norm-diff} (a) with $\widehat{\bm\lambda}_0 = \bm 0$:

We have
\begin{equation}\label{PL.la0.1.S}
\sqrt{n} \leftsecond \nabla g_i(\bm 0) - \nabla g_{i'}(\bm 0) \rightsecond \overset{d}{\to} \mathcal{N}(\bm 0,2\mathbf\Sigma),~\text{as}~n\to\infty
\end{equation}
for $i,i'\in\mathcal{S}^*$ and $i\neq i'$.
The eigenvalues of $2\mathbf\Sigma$ are bounded away from zero and infinity under the assumption.
From \eqref{PL.la0.1.S}, it can be derived
\begin{equation}\label{PL.la0.2.S}
\max_{i,i'\in\mathcal{S}^*,i\neq i'} \| \nabla g_i(\bm 0) - \nabla g_{i'}(\bm 0) \|_2 = O_p(n^{-1/2}\sqrt{\log K})
\end{equation}
according to Lemma \ref{lemma:uni speed.S}.
From \eqref{PL.la0.2.S}, for $\forall \varepsilon>0$, there exists a constant $C$ such that
$$
\Pr \left(\max_{i,i'\in \mathcal S^*, i\ne i'} \|\nabla g_i(\bm 0) - \nabla g_{i'}(\bm 0) \|_2 \le C n^{-1/2} \sqrt{\log K} \right) \ge 1-\varepsilon,
$$
which implies Lemma \ref{lemma:norm-diff} (a).

(ii') Proof of Lemma \ref{lemma:norm-diff} (b) with $\widehat{\bm\lambda}_0 = \bm 0$:

It can be derived
\begin{equation}\label{PL.la0.3.S}
\nabla g_i(\bm 0) - \nabla g_{i'}(\bm 0) = \bm\Delta_{i'} + \frac{1}{n} \sum_{j=1}^n \leftsecond \bm X_{i',j} - \operatorname{E} (\bm X_{i',1}) \rightsecond - \frac{1}{n} \sum_{j=1}^n ( \bm X_{i,j} - \bm\mu_0 ).
\end{equation}
Let $\bm A_{i,i'} = n^{-1} \sum_{j=1}^n \leftsecond \bm X_{i',j} - \operatorname{E} (\bm X_{i',1}) \rightsecond - n^{-1} \sum_{j=1}^n ( \bm X_{i,j} - \bm\mu_0 )$ in \eqref{PL.la0.3.S}. By Central Limit Theorem, we have
\begin{equation}\label{PL.la0.4.S}
\sqrt{n} \bm A_{i,i'} \overset{d}{\to} \mathcal{N}(\bm 0,\mathbf{\Sigma}_{i,i'}),~\text{as}~n\to\infty,
\end{equation}
where $\mathbf{\Sigma}_{i,i'} := \mathbf{\Sigma}_i + \mathbf{\Sigma}_{i'}$.
The eigenvalues of $\mathbf{\Sigma}_{i,i'}$ are bounded away from zero and infinity under the condition of Lemma \ref{lemma:norm-diff}.
Therefore, based on \eqref{PL.la0.4.S}, we can derive
\begin{equation}\label{PL.la0.5.S}
\max_{i\in\mathcal{S}^*,i'\notin\mathcal{S}^*} \left\| \bm A_{i,i'} \right\|_2 = O_p(n^{-1/2}\sqrt{\log K})
\end{equation}
according to Lemma \ref{lemma:uni speed.S}.
Hence, by the assumption $(\log K)^{-1/2} \sqrt{n} \{ \min_{i'\notin\mathcal{S}^*}\| \bm \Delta_{i'} \|_2 - C_n \} \to \infty$ as $n\to\infty$ and \eqref{PL.la0.5.S}, we have
\begin{equation*}
\begin{aligned}
& \Pr \left( \min_{i\in\mathcal{S}^*,i'\notin \mathcal{S}^*} \| \nabla g_i(\bm 0) - \nabla g_{i'}(\bm 0)\|_2 \ge C_n \right)\\
\ge& \Pr \left( \min_{i\in\mathcal{S}^*,i'\notin \mathcal{S}^*} \| \bm \Delta_{i'} \|_2 - C_n \ge \max_{i\in\mathcal{S}^*,i'\notin\mathcal{S}^*} \left\| \bm A_{i,i'} \right\|_2 \right)\\
\to& 1,~\text{as}~n\to\infty.
\end{aligned}
\end{equation*}
This completes the proof of Lemma \ref{lemma:norm-diff} (b).

\subsection{Proof of Theorem \ref{thm:cluster}}

For $\gamma_n$ satisfying $(\log K)^{-1/2} \sqrt{n} (\min_{i'\notin \mathcal{S} ^*}  \| \bm\Delta_{i'} \|_2 - \gamma_n ) \to \infty$ as $n\to \infty$, by Lemma \ref{lemma:norm-diff} (b), we have
\begin{equation*}
\Pr \left( \min_{i\in\mathcal{S}^*,i'\notin \mathcal{S}^*} \| \nabla g_i(\bm\lambda;\bm\mu_0)  - \nabla g_{i'}(\bm\lambda;\bm\mu_0)\|_2 \ge \gamma_n \right) \to 1, ~\text{as}~ n\to \infty.
\end{equation*}
According to Lemma S1 in the supplementary material of \cite{wang2023robust}, we have $\| \widetilde{\bm\mu} - \bm\mu_0 \|_2 = O_p(n^{-1/2}\sqrt{\log K})$ under the assumption $|\mathcal{S}^*|>K/2$.
From Taylor's expansion, we have
$$
\leftsecond \nabla g_i(\bm\lambda;\widetilde{\bm\mu})  - \nabla g_{i'}(\bm\lambda;\widetilde{\bm\mu}) \rightsecond - \leftsecond \nabla g_i(\bm\lambda;\bm\mu_0)  - \nabla g_{i'}(\bm\lambda;\bm\mu_0) \rightsecond = \frac{2}{n} \sum_{j=1}^n (\bm X_{i,j} - \bm X_{i',j})(\widetilde{\bm\mu} - \bm\mu_0)^\top.
$$
We have $\| 2n^{-1} \sum_{j=1}^n (\bm X_{i,j} - \bm X_{i',j})(\widetilde{\bm\mu} - \bm\mu_0)^\top \|_2 = O_p(n^{-1/2}\sqrt{\log K})$, which is negligible comparing to $\nabla g_i(\bm\lambda;\bm\mu_0)  - \nabla g_{i'}(\bm\lambda;\bm\mu_0)$.
Hence, we have
\begin{equation}\label{PT2.eq1.S}
\Pr \left( \min_{i\in\mathcal{S}^*,i'\notin \mathcal{S}^*} \| \nabla g_i(\bm\lambda;\widetilde{\bm\mu})  - \nabla g_{i'}(\bm\lambda;\widetilde{\bm\mu})\|_2 \ge \gamma_n \right) \to 1, ~\text{as}~ n\to \infty.
\end{equation}
Under the assumption $|\mathcal{S}^{*}|> K/2$, for any $i'\notin \mathcal{S}^*$, we have
\begin{equation}\label{PT2.eq2.S}
\Pr (s_{i'} < K/2) \to 1,~\text{as}~n\to\infty
\end{equation}
by \eqref{PT2.eq1.S} and the definition $s_i:= \# \{i':\, \| \nabla g_{i'}(\widehat{\bm\lambda}_0;\widetilde{\bm\mu}) - \nabla g_i(\widehat{\bm\lambda}_0;\widetilde{\bm\mu})\|_2 < \gamma_n, i'=1,\ldots,K \}$.
For any $i'\notin \mathcal{S}^*$, we can prove $\Pr (i'\notin \mathcal{S}) \to 1$, as $n\to\infty$, by \eqref{PT2.eq2.S} and $\mathcal{S} := \{i:\,s_i> K/2,i=1,\ldots,K\}$.
Hence, we have
\begin{equation}\label{PT2.1.S}
\Pr (\mathcal S^* \supseteq \mathcal S)\to 1,~\text{as}~n\to\infty.
\end{equation}

By assumption that $\gamma_n$ satisfies $\gamma_n\sqrt{n}/\sqrt{\log K}\to\infty$ as $n\to \infty$, according to Lemma \ref{lemma:norm-diff} (a), similarly, we have
\begin{equation*}
\Pr \left( \max_{i, i'\in \mathcal{S}^*, i\neq i'}\| \nabla g_i(\bm\lambda;\widetilde{\bm\mu})  - \nabla g_{i'}(\bm\lambda;\widetilde{\bm\mu})\|_2 < \gamma_n \right) \to 1, ~\text{as}~ n\to \infty.
\end{equation*}
For any $i\in \mathcal{S}^*$, it can be derived
\begin{equation}\label{PT2.eq4.S}
\Pr (s_{i} > K/2) \to 1,~\text{as}~n\to\infty
\end{equation}
by the assumption $|\mathcal{S}^{*}|> K/2$ and the definition of $s_i$.
Therefore, we have $\Pr (i\in \mathcal{S}) \to 1$ as $n\to\infty$ 
for any $i\in \mathcal{S}^*$ by \eqref{PT2.eq4.S} and the definition of $\mathcal{S}$.
This implies
\begin{equation}\label{PT2.2.S}
\Pr (\mathcal S^* \subseteq \mathcal S)\to 1,~\text{as}~n\to\infty.
\end{equation}

Combining \eqref{PT2.1.S} and \eqref{PT2.2.S}, we have $\Pr (\mathcal S^* = \mathcal S)\to 1~\text{as}~n\to\infty$, and thus the selection consistency property is proved.

\subsection{Proof of Theorem \ref{thm:RDEL}}

Theorem \ref{thm:DEL} and Theorem \ref{thm:cluster} together proves Theorem \ref{thm:RDEL}.

From Theorem \ref{thm:cluster}, it is known that $\Pr (\mathcal{S} = \mathcal{S}^*) \to 1$ as $n\to \infty$ under the conditions of Lemma \ref{Lemma:assumption} and Theorem \ref{thm:cluster}.
Furthermore, we have $\|\widehat{\bm\lambda}_{\mathcal{S}T} - \widehat{\bm\lambda}_{\mathcal{S}^*} \|_2=o_p(|\mathcal{S}^*|^{-1/2}n^{-1/2})$ for $T \geq \lfloor \log |\mathcal{S}^*|/\log n \rfloor + 1$ by similar techniques to the proof of Theorem \ref{thm:DEL} (a).
Recalling the assumption that $K/2 <|\mathcal{S}^*| \le K$, we have
$$
\lfloor \log K/\log n \rfloor + 1 \ge \lfloor \log |\mathcal{S}^*|/\log n \rfloor + 1.
$$
Hence, it can be derived
$$
\|\widehat{\bm\lambda}_{\mathcal{S}T} - \widehat{\bm\lambda}_{\mathcal{S}^*}\|_2 = o_p(|\mathcal{S}^*|^{-1/2}n^{-1/2})= o_p(N^{-1/2}),
$$
for $T\ge \lfloor \log K/\log n \rfloor + 1$.
This completes the proof of Theorem \ref{thm:RDEL} (a).

Similar to the proof of Theorem \ref{thm:DEL} (b), we have
$$
2\sum_{i\in\mathcal{S}}\sum_{j=1}^n \log\leftsecond 1+ \widehat{\bm\lambda}_{\mathcal{S}T} (\bm X_{i,j} - \bm\mu_0) \rightsecond \overset{d}{\to} \chi^2_{(d)}
$$
from $\Pr (\mathcal S = \mathcal S^*)\to 1$ as $n\to\infty$ and $\|\widehat{\bm\lambda}_{\mathcal{S}T} - \widehat{\bm\lambda}_{\mathcal{S}^*} \|_2 = o_p(N^{-1/2})$.
This completes the proof of Theorem \ref{thm:RDEL} (b).

\vskip 0.2in
\bibliography{DEL}

\end{document}